\documentclass[letterpaper,11pt]{article}

\usepackage[T1]{fontenc}
\usepackage[margin=1in]{geometry}
\usepackage{mathptmx}
\DeclareMathAlphabet{\mathcal}{OMS}{cmsy}{m}{n}

\usepackage{enumitem}
\setlist{noitemsep,topsep=0pt,parsep=0pt,partopsep=0pt}
\usepackage[small,compact]{titlesec}
\usepackage[small,it,skip=5pt]{caption}

\usepackage{calc}
\usepackage{authblk}
\usepackage{float,xspace}
\usepackage{amsmath, amsthm, amssymb}
\usepackage[hidelinks]{hyperref}
\usepackage[nottoc, notlof, notlot]{tocbibind}

\usepackage{accents}
\newcommand*{\dt}[1]{\accentset{\mbox{\large\bfseries .}}{#1}}
\newcommand*{\ddt}[1]{\accentset{\mbox{\large\bfseries .\hspace{-0.05ex}.}}{#1}}

\newtheorem{thm}{Theorem}[section]
\newtheorem{prop}[thm]{Proposition}
\newtheorem{cor}[thm]{Corollary}
\newtheorem{lem}[thm]{Lemma}
\newtheorem{fact}[thm]{Fact}
\newtheorem{defn}[thm]{Definition}
\newtheorem{rem}[thm]{Remark}
\newtheorem{nota}[thm]{Notations}
\theoremstyle{definition}

\newcommand{\boxalgo}[3]{
\renewcommand{\figurename}{Algorithm}
  \begin{figure}[H]
    \centering
    \framebox{\begin{minipage}{\textwidth-1cm} \small #3 \end{minipage}}
    \caption{#2}
    \label{#1}
  \end{figure}
\renewcommand{\figurename}{Figure}}

\newcommand{\boxest}[3]{
\renewcommand{\figurename}{Estimator}
  \begin{figure}[H]
    \centering
    \framebox{\begin{minipage}{\textwidth-1cm} \small #3 \end{minipage}}
    \caption{#2}
    \label{#1}
  \end{figure}
\renewcommand{\figurename}{Figure}}

\newcommand{\R}{\mathbb{R}}

\newcommand{\C}{\mathbb{C}}
\newcommand{\ra}{\rightarrow}
\newcommand{\rn}{\{0,1\}}
\newcommand{\eps}{\epsilon}
\def\polylog{\operatorname{polylog}}
\def\median{\operatorname{median}}
\newcommand{\poly}{\mathit{poly}}

\newcommand{\bo}[1]{\mathcal{O}\left(#1\right)}
\newcommand{\Xso}{\widetilde{\mathcal{O}}}
\newcommand{\so}[1]{\Xso\left(#1\right)}
\newcommand{\om}[1]{\Omega\left(#1\right)}
\newcommand{\thet}[1]{\Theta\left(#1\right)}
\newcommand{\sthet}[1]{\widetilde{\Theta}\left(#1\right)}

\newcommand{\Pb}{\mathbb{P}} 
\newcommand{\esp}[1]{\mathbb{E}\left[#1\right]} 
\newcommand{\var}[1]{\mathrm{Var}\left[#1\right]} 
\newcommand{\ch}{\Delta_{\samp}} 
\newcommand{\hi}{H} 
\newcommand{\lo}{L} 
\newcommand{\md}{M} 
\newcommand{\wid}{\Gamma} 

\newcommand{\Hil}{\mathcal{H}}
\newcommand{\Hqubit}{\mathbb{C}^{2}}
\newcommand{\spa}[1]{\textup{Span}\left(#1\right)}
\newcommand{\bra}[1]{\langle #1|}
\newcommand{\ket}[1]{| #1 \rangle}
\newcommand{\bracket}[3]{\langle #1|#2|#3 \rangle}
\newcommand{\ip}[2]{\langle #1|#2 \rangle}
\newcommand{\proj}[1]{| #1 \rangle \langle #1 |}

\newcommand{\Xamp}{\textup{\textsf{AmplEst}}} 
\newcommand{\amp}[1]{\Xamp\left(#1\right)}
\newcommand{\Xamps}{\textup{\textsf{AmplEst}$^\star$}} 
\newcommand{\amps}[1]{\Xamps\left(#1\right)}
\newcommand{\Xamplify}{\textup{\textsf{Amplify}}} 
\newcommand{\amplify}[1]{\Xamplify\left(#1\right)}
\newcommand{\Xampb}{\textup{\textsf{BasicEst}}} 
\newcommand{\ampb}[1]{\Xampb\left(#1\right)}

\newcommand{\qs}{quantum sample\xspace}
\newcommand{\qss}{quantum samples\xspace}
\newcommand{\samp}{\mathcal{S}}
\newcommand{\mus}{\mu_{\samp}} 
\newcommand{\muts}{\widetilde{\mu}_{\samp}}
\newcommand{\phis}{\phi_{\samp}} 
\newcommand{\phits}{\widetilde{\phi}_{\samp}}
\newcommand{\sigs}{\sigma_{\samp}} 

\newcommand{\alg}{\mathcal{A}}
\newcommand{\balg}{\mathcal{B}}

\newcommand{\talgs}{T_{\ell_2}}
\newcommand{\tmax}{T_{\mathit{max}}}
\newcommand{\sto}{\mathit{stop}}
\newcommand{\cont}{\mathit{cont}}
\newcommand{\acc}{\mathit{acc}}
\newcommand{\rej}{\mathit{rej}}
\newcommand{\gena}[1]{\textup{\textsf{Gen}}_{\alg}\left(#1\right)}
\newcommand{\genb}[1]{\textup{\textsf{Gen}}_{\balg}\left(#1\right)}

\newcommand{\str}{\vec{u}} 
\newcommand{\astr}{\mathcal{T}} 

\newcommand{\mut}{\widetilde{\mu}}

\newcommand{\at}{\widetilde{a}}
\newcommand{\bt}{\widetilde{b}}
\newcommand{\pt}{\widetilde{p}}
\newcommand{\qt}{\widetilde{q}}
\newcommand{\mt}{\widetilde{m}}
\newcommand{\tti}{\widetilde{t}}

\newcommand{\bucket}{I^+}
\newcommand{\disj}{\textsc{Disjointness}}

\title{Quantum Chebyshev's Inequality and Applications}

\author{Yassine Hamoudi}
\author{Fr\'{e}d\'{e}ric Magniez}
\affil{IRIF, Universit\'{e} Paris Diderot, CNRS, France\\\texttt{\{hamoudi,magniez\}@irif.fr}}

\date{}

\begin{document}

\maketitle
\thispagestyle{empty}
\setcounter{page}{0}


\begin{abstract}
  In this paper we provide new quantum algorithms with polynomial speed-up for a range of problems for which no such results were known, or we improve previous algorithms. First, we consider  the approximation of the frequency moments $F_k$ of order $k \geq 3$ in the multi-pass streaming model with updates (turnstile model). We design a $P$-pass quantum streaming algorithm with memory $M$ satisfying a tradeoff of $P^2 M = \so{n^{1-2/k}}$, whereas the best classical algorithm requires $P M = \Theta(n^{1-2/k})$. Then, we study the problem of estimating the number $m$ of edges and the number $t$ of triangles given query access to an $n$-vertex graph. We describe optimal quantum algorithms that perform $\so{\sqrt{n}/m^{1/4}}$ and $\so{\sqrt{n}/t^{1/6} + m^{3/4}/\sqrt{t}}$ queries respectively. This is a quadratic speed-up compared to the classical complexity of these problems.

  For this purpose we develop a new quantum paradigm that we call Quantum Chebyshev's inequality. Namely we demonstrate that, in a certain model of quantum sampling, one can approximate with \emph{relative error} the mean of any random variable with a number of \qss that is linear in the ratio of the square root of the variance to the mean. Classically the dependency is quadratic. Our algorithm subsumes a previous result of Montanaro \cite{Mon15}. This new paradigm is based on a refinement of the Amplitude Estimation algorithm of Brassard et al. \cite{BHMT02} and of previous quantum algorithms for the mean estimation problem. We show that this speed-up is optimal, and we identify another common model of quantum sampling where it cannot be obtained. For our applications, we also adapt the variable-time amplitude amplification technique of Ambainis \cite{Amb10b} into a variable-time amplitude estimation algorithm.
\end{abstract}

\newpage


\section{Introduction}

\paragraph{Motivations and background}
Randomization and probabilistic methods are among the most widely used techniques in modern science, with applications ranging from mathematical economics to medicine or particle physics. One of the most successful probabilistic approaches is the Monte Carlo Simulation method for algorithm design, that relies on repeated random sampling and statistical analysis to estimate parameters and functions of interest. From Buffon's needle experiment, in the eighteenth century, to the simulations of galaxy formation or nuclear processes, this method and its variations have become increasingly popular to tackle problems that are otherwise intractable. The Markov chain Monte Carlo method \cite{JS96} led for instance to significant advances for approximating parameters whose exact computation is \textsc{\#P}-hard \cite{KL83,JVV86,DFK91,JSV04}.

The analysis of Monte Carlo Simulation methods is often based on concentration inequalities that characterize the deviation of a random variable from some parameter. In particular, the Chebyshev inequality is a key element in the design of randomized methods that estimate some target numerical value. Indeed, this inequality guarantees that the arithmetic mean of $\Delta^2/\eps^2$ independent samples, from a random variable with variance $\sigma^2$ and mean $\mu$ satisfying $\Delta \geq \sigma/\mu$, is an approximation of $\mu$ under relative error $\eps$ with high probability. This basic result is at the heart of many computational problems, such as counting via Markov chains \cite{JS96,SVV09}, estimating graph parameters \cite{CRT05,Fei06,GR08,ELRS17}, testing properties of classical \cite{GR11,BFRSW13,CDVV14,CDKS18} or quantum \cite{BHH11,BOW17} distributions, approximating the frequency moments in the data stream model \cite{AMS99,MW10,AKO11}.

Various quantum algorithms have been developed to speed-up or generalize classical Monte Carlo methods (e.g. sampling the stationary distributions of Markov-chains \cite{WA08,PW09,DGG10,TOVP11,CS17}, estimating the expected values of observables or partition functions \cite{KOS07,WCNA09,PW09,Mon15}). The mean estimation problem (as addressed by Chebyshev's inequality) has also been studied in the \emph{quantum sampling model}. In this model, a distribution is represented by a unitary transformation (called a \emph{quantum sampler}) preparing a superposition over the elements of the distribution, with the amplitudes encoding the probability mass function. A \emph{\qs} is defined as one execution of a quantum sampler or its inverse. The number of \qss needed to estimate the mean of a distribution on a bounded space $[0,B]$, with \emph{additive} error $\eps$, was proved to be $\bo{B/\eps}$ \cite{Hei02,BDGT11}, or $\so{\bar{\sigma}/\eps}$ \cite{Mon15} given an upper-bound $\bar{\sigma}^2$ on the variance. On the other hand, the mean estimation problem with \emph{relative} error $\eps$ can be solved with $\bo{\sqrt{B}/(\eps\sqrt{\mu})}$ \qss \cite{BHMT02,WCNA09}. Interestingly, this is a quadratic improvement over $\sigma^2/(\eps \mu)^2$ if the sample space is $\{0,B\}$ (this case maximizes the variance). Montanaro \cite{Mon15} posed the problem of whether this speed-up can be generalized to other distributions. He assumed that one knows an upper bound\footnote{More precisely, $\Delta$ is an upper bound on $\phi/\mu$ where $\phi^2$ is the second moment, which satisfies $\sigma/\mu\leq \phi/\mu\leq 1+\sigma/\mu$.} $\Delta$ on $1+\sigma/\mu$, and gave an algorithm using\footnote{We use the notation $\so{x}$ to indicate $\bo{x \cdot \polylog x}$.} $\so{\Delta^2/\eps}$ \qss (thus improving the dependence on $\eps$, compared to the classical setting). This result was reformulated in \cite{LW17} to show that, knowing bounds $\lo\leq \mu\leq\hi$, it is possible to use $\so{\Delta/\eps \cdot \hi/\lo}$ \qss. Typically, the only upper-bound known on $\mu$ is $\hi = B$, so it is less efficient than \cite{BHMT02,WCNA09}.

\paragraph{Quantum Chebyshev Inequality}
Our main contribution (Theorem \ref{Thm:basicEpsApprox} and Theorem \ref{Thm:EpsApprox}) is to show that the mean $\mu$ of any distribution with variance $\sigma^2$ can be approximated with relative error $\eps$ using $\so{\Delta \cdot \log(\hi/\lo) + \Delta/\eps}$ \qss, given an upper bound $\Delta$ on $1+\sigma/\mu$ and two bounds $\lo,\hi$ such that  $\lo < \mu < \hi$. This is an exponential improvement in $\hi/\lo$ compared to previous works \cite{LW17}. Moreover, if $\log(\hi/\lo)$ is negligible, this is a quadratic improvement over the number of classical samples needed when using the Chebyshev inequality. If no  bound $\lo$ is known, we also present an algorithm using $\so{\Delta/\eps \cdot \log^3(\hi/\mu)}$ \qss in expectation (Theorem \ref{Thm:ExpectedEpsApprox}). A corresponding lower bound is deduced from \cite{NW99} (Theorem~\ref{Thm:lowerCheb}). We also show (Theorem \ref{Thm:lowerState}) that no such speed-up is possible if we only had access to \emph{copies} of the quantum state representing the distribution.

Our algorithm is based on \emph{sequential analysis}. Given a threshold $b \geq 0$, we will consider the ``truncated'' mean $\mu_{< b}$ defined by replacing the outcomes larger than $b$ with $0$. Using standard techniques, this mean can be encoded in the amplitude of some quantum state $\sqrt{1-\mu_{< b}/b} \ket{\psi} + \sqrt{\mu_{< b}/b} \ket{\psi^{\perp}}$ (Corollary \ref{Cor:BasicEst}). We then run the \emph{Amplitude Estimation} algorithm of Brassard et al. \cite{BHMT02} on this state for $\Delta$ steps (i.e. with $\Delta$ \qss), only to see whether the estimate of $\mu_{< b}/b$ it returns is nonzero (this is our \emph{stopping rule}). A property of this algorithm (Corollary \ref{Cor:BasicEst} and Remark \ref{Obs:zero}) guarantees that it is zero with high probability if and only if the number of \qss is below the inverse $\sqrt{b/\mu_{< b}}$ of the estimated amplitude. The crucial observation (Lemma \ref{Lem:mdRange}) is that $\sqrt{b/\mu_{< b}}$ is smaller than $\Delta$ for large values of $b$, and it becomes larger than $\Delta$ when $b \approx \mu \Delta^2$. Thus, by repeatedly running the amplitude estimation algorithm with $\Delta$ \qss, and doing $\bo{\log(\hi/\lo)}$ steps of a logarithmic search on decreasing values of $b$, the first non-zero value is obtained when $b/\Delta^2$ is approximately equal to $\mu$. The precision of the result is later improved, by using more precise ``truncated'' means.

This algorithm is extended (Theorem \ref{Thm:EpsApproxDec}) to cover the common situation where one knows a non-increasing function $f$ such that $f(\mu) \geq 1+\sigma/\mu$, instead of having explicitly $\Delta \geq 1+\sigma/\mu$. For this purpose, we exhibit another property (Corollary \ref{Cor:BasicEst} and Remark \ref{Obs:Markov}) of the amplitude estimation algorithm, namely that it always outputs a number smaller than the estimated value (up to a constant factor) with high probability. This shall be seen as a quantum equivalent of the Markov inequality. Combined with the previous algorithm, it allows us to find a value $f(\mut) \geq 1+\sigma/\mu$, with a second logarithmic search on $\mut$.

Next, we study the quantum analogue of the following standard fact: $s$ classical samples, each taking average time $T_{av}$ to be computed, can be obtained in total average time $s \cdot T_{av}$. The notion of average time is adapted to the quantum setting, using the framework of variable-time algorithms introduced by Ambainis. We develop a variable-time amplitude estimation algorithm (Theorem \ref{Thm:VarTime}) that approximates the target value efficiently when some branches of the computation stop earlier than the others. It can be used in place of the standard amplitude estimation in all our results (Theorem \ref{Thm:VarEpsApprox}).

\paragraph{Applications}
We describe two applications that illustrate the use of the above results. We first study the problem of approximating the frequency moments $F_k$ of order $k \geq 3$ in the multi-pass streaming model with updates. Classically, the best $P$-pass algorithms with memory $M$ satisfy $P M = \thet{n^{1-2/k}}$ \cite{MW10,WZ12}. We give a quantum algorithm  for which $P^2 M = \so{n^{1-2/k}}$ (Theorem~\ref{Thm:fkQuant}). This problem was studied before in \cite{Mon16}, where the author obtained quantum speed-ups for $F_0$, $F_2$ and $F_{\infty}$, but no significant improvement for $k \geq 3$. Similar tradeoff results are known for \textsc{Disjointness} ($P^2 M = \sthet{n}$ in the quantum streaming model \cite{Gal09} vs. $P M = \thet{n}$ classically), and \textsc{Dyck(2)} ($P^3 M = \om{\sqrt{n}}$ \cite{NT17} vs. $P M = \sthet{\sqrt{n}}$  \cite{MMN14,CCKM13,JN14}).

Our construction starts with a classical one-pass \emph{linear sketch} streaming algorithm \cite{MW10,AKO11} with memory $\polylog n$, that samples (approximately) from a distribution with mean $F_k$ and variance $\bo{n^{1-2/k} F_k^2}$. We implement it with a quantum sampler, that needs two passes for one \qs. The crucial observation (Appendix \ref{App:appFrequencyMoments}) is that the reverse computation of a linear sketch algorithm can be done efficiently in one pass (whereas usually that would require processing the same stream but in the reverse direction).

As a second application, we study the approximation of graph parameters using neighbor, vertex-pair and degree queries.  We show that the numbers $m$ of edges and $t$ of triangles, in an $n$-vertex graph, can be estimated with $\sthet{n^{1/2}/m^{1/4}}$ (Theorem \ref{Thm:edge}) and $\sthet{\sqrt{n}/t^{1/6} + m^{3/4}/\sqrt{t}}$ (Theorem \ref{Thm:triangleEps}) quantum queries respectively. This is a quadratic speed-up over the best classical algorithms \cite{GR08,ELRS17}. The lower bounds (Theorems \ref{Thm:edgeLB} and \ref{Thm:triangleLB}) are obtained with a property testing to communication complexity reduction method.

The number of edges is approximated by translating a classical estimator \cite{Ses15} into a quantum sampler. The triangle counting algorithm is more involved. We need a classical estimator \cite{ELRS17} approximating the number $t_v$ of adjacent triangles to any vertex $v$. Its average running time being small, we obtain a quadratic speed-up for estimating $t_v$ (Proposition \ref{Prop:quantTv}) using our mean estimation algorithm for variable-time samplers. We then diverge from the classical triangle counting algorithm of \cite{ELRS17}, that requires to set up a data structure for sampling edges uniformly in the graph. This technique seems to be an obstacle for a quadratic speed-up. We circumvent this problem by adapting instead a bucketing approach from \cite{ELR15} that partitions the graph's vertices according to the value of $t_v$. The size of each bucket is estimated using a second quantum sampler.

\section{Preliminaries}

\subsection{Computational model}
\label{Sec:model}

In this paper we consider probability distributions $d$ on some finite sample spaces $\Omega \subset \R^+$.
We denote by $d(x)$ the probability to sample $x\in \Omega$ in the distribution $d$.
We also make the assumption, which is satisfied for most of applications, that $\Omega$ is equipped with an efficient encoding of its elements $x\in\Omega$. In particular, we can perform quantum computations on the Hilbert space $\Hil_{\Omega}$ defined by the basis $\{\ket{x}\}_{x \in \Omega}$.
Moreover, given any two values $0 \leq a < b$, we assume the existence of a unitary $R_{a,b}$ that can perform the \emph{Bernoulli sampling} (see below) in time polylogarithmic in $b$. In the rest of the paper we will neglect this complexity, including the required precision for implementing any of those unitary operators.

\begin{defn}
  \label{Def:Bernoulli}
  Given a finite  space $\Omega \subset \R^+$ and two reals $0 \leq a < b$, an \emph{$(a,b)$-Bernoulli sampler over $\Omega$} is  a unitary $R_{a,b}$ acting on $\Hil_{\Omega} \otimes \Hqubit$ and satisfying for all $x \in \Omega$:
    \begin{align*}
      R_{a,b} (\ket{x}\ket{0}) =
        \begin{cases}
          \ket{x} \left(\sqrt{1-\frac{x}{b}}\ket{0} + \sqrt{\frac{x}{b}} \ket{1}\right) & \text{when $a \leq x < b$,} \\
          \ket{x}\ket{0} & \text{otherwise.}
        \end{cases}
    \end{align*}
    We say that $\Omega$ is \emph{Bernoulli samplable} if any $(a,b)$-Bernoulli sampler can be implemented
    in polylogarithmic time in ${b}$, when $a,b$ have polylog-size encodings in $b$.
\end{defn}

The $R_{a,b}$ operation can be implemented with a controlled rotation, and is reminiscent of related works on mean estimation (e.g. \cite{WCNA09,BDGT11,Mon15}). In what follows, we always use $a = 0$ or $a = b/2$.

We can now define what a \emph{quantum sample} is.

\begin{defn}
  \label{Def:quantSampl}
  Given a finite Bernoulli samplable space  $\Omega \subset \R^+$ and a distribution $d$ on $\Omega$,
  a \emph{(quantum) sampler} $\samp$ for $d$ is a unitary operator acting on $\Hil_{g} \otimes \Hil_{\Omega}$, for some Hilbert space $\Hil_{g}$, such that
    \[\samp (\ket{0} \ket{0}) = \sum_{x \in \Omega} \sqrt{d(x)} \ket{\psi_x} \ket{x}\]
  where $\ket{\psi_x}$ are arbitrary unit vectors. A \emph{\qs} is one execution of $\samp$ or $\samp^{-1}$ (including their controlled versions). The \emph{output} of $\samp$ is the random variable $v(\samp)$ obtained by measuring the $x$-register of $\samp (\ket{0} \ket{0})$. Its \emph{mean} is denoted by $\mus$, its \emph{variance} by $\sigs^2 $, and its \emph{second moment}  by $\phis^2 = \esp{v(\samp)^2}$.
\end{defn}

Given a non-negative random variable $X$ and two numbers $0 \leq a \leq b$, we define the random variable $X_{a,b} = \mathrm{id}_{a,b}(X)$ where $\mathrm{id}_{a,b}(x) = x$ when $a \leq x < b$ and $\mathrm{id}_{a,b}(x) = 0$ otherwise. If $a = 0$, we let $X_{< b} =  X_{0,b}$. Similarly, $X_{\geq b} = \mathrm{id}_{\geq b}(X)$ where $\mathrm{id}_{\geq b}(x) = x$ when $x \geq b$ and $\mathrm{id}_{\geq b}(x) = 0$ otherwise.

We motivate the use of a Bernoulli sampler $R_{a,b}$ by the following observation: for any sampler $\samp$ and values $0 \leq a < b$, the modified sampler $\hat{\samp} = (I_{\Hil_{g}} \otimes R_{a,b}) (\samp \otimes I_{\C^2})$ acting on $\Hil_{\hat{g}} \otimes \Hil_{\hat{\Omega}}$, where $\Hil_{\hat{g}} = \Hil_{g} \otimes \Hil_{\Omega}$ and $\hat{\Omega} = \rn$, generates the Bernoulli distribution $d(0) = 1 - p$, $d(1) = p$ of mean $p = \esp{v(\hat{\samp})} = b^{-1} \esp{v(\samp)_{a,b}}$ (see the proof of Corollary \ref{Cor:BasicEst}). This central result will be used all along this paper.

\paragraph{Other quantum sampling models}
Instead of having access to the unitary $\samp$, one could only have copies of the state $\sum_{x \in \Omega} \sqrt{d(x)} \ket{\psi_x} \ket{x}$ (as in \cite{AW17} for instance). However, as we show in Theorem \ref{Thm:lowerState}, the speed-up presented in this paper is impossible to achieve in this model. On another note, Aharonov and Ta-Shma \cite{AT07} studied the $\emph{Qsampling}$ problem, which is the ability to prepare $\sum_{x \in \Omega} \sqrt{d(x)} \ket{x}$ given the decription of a classical circuit with output distribution $d$. This problem becomes straightforward if a garbage register $\psi_x$ can be added (using standard reversible-computation techniques). Bravyi, Harrow and Hassidim \cite{BHH11} considered an oracle-based model, that is provably weaker than Qsampling, where a distribution $d = (d(1),\dots,d(N))$ on $\Omega = [N]$ is represented by an oracle $O_d : [S] \ra [N]$ (for some $S$), such that $d(x)$ equals the proportion of inputs $s \in [S]$ with $O_d(s) = x$. It is extended to the quantum query framework with a unitary $\mathcal{O}_d$ such that $\mathcal{O}_d \ket{s} \ket{0} = \ket{s} \ket{O_d(s)}$. It is not difficult to see that applying $\mathcal{O}_d$ on a uniform superposition gives $\sum_{x \in [N]} \sqrt{d(x)} \left(\frac{1}{\sqrt{d(x) S}} \sum_{s \in [S] : O_d(s) = x} \ket{s}\right) \ket{x}$, as required by Definition \ref{Def:quantSampl} (where $\ket{\psi_x} = \frac{1}{\sqrt{d(x) S}} \sum_{s \in [S] : O_d(s) = x} \ket{s}$). Finally, Montanaro \cite{Mon15} presented a model that is similar to ours, where he replaced the $x$-register of $\samp (\ket{0} \ket{0})$ with a $k$-qubit register (for some $k$) combined with a mapping $\phi : \rn^k \ra \Omega$ where $x = \phi(s)$ is the sample associated to each $s \in \rn^k$.


  \subsection{Amplitude estimation}
  \label{Sec:amplEst}

The essential building block of this paper is the amplitude estimation algorithm \cite{BHMT02}, combined with ideas from \cite{WCNA09,BDGT11,Mon15}, to estimate the modified mean $b^{-1} \esp{v(\samp)_{a,b}}$ of a quantum sampler $\samp$ to which a Bernoulli sampler $R_{a,b}$ has been applied. We will need the following result about amplitude estimation.

\begin{thm}
  \label{Thm:amplEst}
  There is a quantum algorithm $\Xamp$, called \emph{Amplitude Estimation}, that takes as input a unitary operator $U$, an orthogonal projector $\Pi$, and an integer $t > 2$. The algorithm outputs an estimate $\pt = \amp{U,\Pi,t}$ of $p = \bra{\psi} \Pi \ket{\psi}$, where $\ket{\psi} = U \ket{0}$, such that
  $$
    \begin{cases}
     |\pt-p| \leq 2\pi \frac{\sqrt{p}}{t} + \frac{\pi^2}{t^2},  &\text{with probability $8/\pi^2$;} \\
      \pt = 0 ,                                                & \text{with probability $\frac{\sin^2(t \theta)}{t^2 \sin^2(\theta)}$.}
    \end{cases}
  $$
  and $0 \leq \theta \leq \pi/2$ satisfies $\sin(\theta) = \sqrt{p}$. It uses $\bo{\log^2(t)}$ $2$-qubit quantum gates (independent of $U$ and $\Pi$) and makes $2t+1$ calls to (the controlled versions of) $U$ and $U^{-1}$, and $t$ calls to the reflection $I - 2\Pi$.
\end{thm}

We now present an adaptation of the algorithms from~\cite{WCNA09,BDGT11,Mon15} for estimating $b^{-1} \esp{v(\samp)_{a,b}}$.

\boxalgo{Algo:BasicEst}{the Basic Estimation algorithm $\Xampb$.}{
{\bf Input:} a sampler $\samp$ acting on $\Hil_{g} \otimes \Hil_{\Omega}$, two values $(a,b)$, an integer $t$, a failure parameter $0 < \delta < 1$. \\
{\bf Output:} an estimate $\pt = \ampb{\samp,(a,b),t,\delta}$ of $p = b^{-1}{\esp{v(\samp)_{a,b}}}$

\begin{enumerate}
  \item Let $U = (I_{\Hil_{g}} \otimes R_{a,b}) (\samp \otimes I_{\C^2})$ and $\Pi = I_{\Hil_{g}} \otimes I_{\Hil_{\Omega}} \otimes \proj{1}$.
  \item For $i = 1,\dots,\thet{\log(1/\delta)}$: compute $\pt_i = \amp{U,\Pi,t}$.
  \item \underline{Output} $\pt = \median\{\pt_1,\dots,\pt_{\thet{\log(1/\delta)}}\}$.
\end{enumerate}
}

\begin{cor}
  \label{Cor:BasicEst}
  Consider a quantum sampler $\samp$ and two values $0 \leq a < b$. Denote $p = b^{-1}{\esp{v(\samp)_{a,b}}}$.
  Given an integer $t > 2$ and a real $0 < \delta < 1$, $\ampb{\samp,(a,b),t,\delta}$ (see Algorithm~\ref{Algo:BasicEst})
  uses $\bo{t \log(1/\delta)}$  \qss and outputs $\pt$ satisfying all of the following inequalities with probability $1-\delta$:
  \medskip

   \noindent \centering
   \begin{tabular}{l l l l}
     \emph{(1)} \ $|\pt-p| \leq 2\pi \frac{\sqrt{p}}{t} + \frac{\pi^2}{t^2}$, & for any $t$; & \quad \emph{(2)} \ $\pt \leq (1+2\pi)^2 \cdot p$,  & for any $t$; \\[.2cm]
     \emph{(3)} \ $\pt = 0$, &     when $t < \frac{1}{2 \sqrt{p}}$; & \quad \emph{(4)} \ $|\pt-p| \leq \eps \cdot p$, & when $t \geq \frac{8}{\eps \sqrt{p}}$ and $0 < \eps < 1$.
   \end{tabular}
\end{cor}

\begin{proof}
  We show that each $\pt_i$ satisfies the inequalities stated in the corollary, with probability $8/\pi^2$. Since $\pt$ is the median of $\thet{\log 1/\delta}$ such values, the probability is increased to $1-\delta$ using the Chernoff bound.

  For each $x \in \Omega$, denote $\nu_x = \frac{x}{b}$ if $a \leq x < b$, and $\nu_x=0$ otherwise. Since $p = \sum_{x \in \Omega} \nu_x d(x)$, observe that
    $$
      U (\ket{0} \ket{0} \ket{0})
         = \sum_{x \in \Omega} \sqrt{d(x)} \ket{\psi_x} \ket{x} \left(\sqrt{1-\nu_x}\ket{0} + \sqrt{\nu_x} \ket{1}\right)
         = \sqrt{1-p} \ket{\psi'_0} \ket{0} + \sqrt{p} \ket{\psi'_1} \ket{1}
    $$
  where $\ket{\psi'_0} = \frac{1}{\sqrt{1-p}} \sum_{x \in \Omega} \sqrt{d(x)} \sqrt{1-\nu_x} \ket{\psi_x} \ket{x}$ and $\ket{\psi'_1} = \frac{1}{\sqrt{p}} \sum_{x \in \Omega} \sqrt{d(x)} \sqrt{\nu_x} \ket{\psi_x} \ket{x}$ are unit vectors. Thus, the output $\pt_i$ of the $\Xamp$ algorithm applied on $U$ and $\Pi$ is an estimate of $p$ satisfying the output conditions of Theorem \ref{Thm:amplEst}. Therefore $|\pt_i-p| \leq 2\pi \frac{\sqrt{p}}{t} + \frac{\pi^2}{t^2}$ with probability $8/\pi^2$, for any $t$. By plugging $t \geq \frac{8}{\eps \sqrt{p}}$ into this inequality we have $|\pt_i-p| \leq \eps \cdot p$. By plugging $t \geq \frac{1}{2 \sqrt{p}}$ we also have $|\pt_i-p| \leq (4\pi+4\pi^2) p$, and thus $\pt_i \leq (1+2\pi)^2 \cdot p$. Finally, if $t < \frac{1}{2 \sqrt{p}}$, denote $0 \leq \theta \leq \pi/2$ such that  $\sin(\theta) = \sqrt{p}$ and observe that $\theta \leq \frac{\pi}{2} \sqrt{p} \leq \frac{\pi}{4t}$ (since $\frac{2}{\pi} x \leq \sin(x) \leq x$, for $x \in [0,\pi/2]$).
  The probability to obtain $\pt_i = 0$ is $\frac{\sin^2(t \theta)}{t^2 \sin^2(\theta)} \geq \frac{\sin^2(t \pi/(4t))}{t^2 \sin^2(\pi/(4t))} \geq \frac{\sin^2(\pi/4)}{t^2 (\pi/(4t))^2}  = 8 / \pi^2$, since $x \mapsto \sin^2(tx)/(t^2 \sin^2(x))$ is decreasing for $0 < x \leq \pi/t$. Moreover, when $t < \frac{1}{2 \sqrt{p}}$, the first two inequalities are obviously satisfied if $\pt_i = 0$.
\end{proof}

The four results on $p$ in Corollary \ref{Cor:BasicEst} lie at the heart of this paper. We make a few comments on them.

\begin{rem}
  \label{Obs:bernoulli}
  Consider a sampler $\samp$ over $\Omega = \rn$ for the Bernoulli distribution of parameter $p$. Using the Chebyshev inequality, we get that $\bo{(1-p)/(\eps^2 p)}$ classical samples are enough for estimating $p$ with relative error $\eps$. The inequality (4) of Corollary \ref{Cor:BasicEst} shows that $t = \bo{1/(\eps \sqrt{p})}$ \qss are sufficient. Our main result (Section \ref{Sec:cheb}) generalizes this quadratic speed-up to the non-Bernoulli case.
\end{rem}

\begin{rem}
  \label{Obs:Markov}
  The inequality (2) shall be seen as an equivalent of the Markov inequality\footnote{The Markov inequality for a non-negative random variable $X$ states that $\Pb(X \geq k \esp{X}) \leq 1/k$ for any $k > 0$. Here, although we do not need this result, it is possible to prove that $\Pb(\pt \geq k p) \leq C/\sqrt{k}$, for some absolute constant $C$.}, namely that $\pt$ does not exceed $p$ by a large factor with large probability. This property will be used in Appendix \ref{App:appDecreasingFunction}.
\end{rem}

\begin{rem}
  \label{Obs:zero}
  If $p \neq 0$, inequalities (3) and (4) imply that, with large probability, $t < 8/\sqrt{p}$ when $\pt = 0$, and $t \geq 1/(2\sqrt{p})$ when $\pt \neq 0$. This phenomenon, at $t = \Theta(1/\sqrt{p})$, is crucially used in the next section.
\end{rem}

\section{Quantum Chebyshev's inequality}
\label{Sec:cheb}

We describe our main algorithm for estimating the mean $\mus$ of any quantum sampler $\samp$, given an upper bound $\ch \geq \phis/\mus$ (we recall that $\phis^2 = \esp{v(\samp)^2}$ and $\sigs/\mus\leq \phis/\mus\leq 1+\sigs/\mus$). The two main tools used in this section are the $\Xampb$ algorithm of Corollary \ref{Cor:BasicEst}, and the following lemma on ``truncated'' means. We recall that $X_{< b}$ (resp. $X_{\geq b}$) is defined from a non-negative random variable $X$ by substituting the outcomes greater or equal to $b$ (resp. less than $b$) with $0$. Note that $X = X_{< b} + X_{\geq b}$ for all $b > 0$.

\begin{fact}
  \label{fact:range}
  For any random variable $X$ and numbers $0 < a \leq b$, we have $\esp{X_{a,b}} \leq \frac{\esp{X_{a,b}^2}}{a}$ and $\esp{X_{\geq b}} \leq \frac{\esp{X_{\geq b}^2}}{b}$.
\end{fact}

\begin{lem}
  \label{Lem:mdRange}
  Let $X$ be a non-negative random variable and $\Delta \geq \sqrt{\esp{X^2}}/\esp{X}$. Then, for all $c_1, c_2, M > 0$ such that $c_1 \cdot \esp{X} \leq \md \leq c_2 \cdot \esp{X}$, we have 
    \[\left(1-\frac{1}{c_1}\right) \cdot \esp{X} \leq \esp{X_{< \md \Delta^2}} \leq \esp{X}
      \qquad \text{and} \qquad
      \sqrt{c_1} \cdot \Delta \leq \frac{1}{\sqrt{\esp{X_{< \md \Delta^2}}/(\md \Delta^2)}} \leq \sqrt{c_2 \left(1-\frac{1}{c_1}\right)} \cdot \Delta\]
\end{lem}

\begin{proof}
  The left hand side term is a consequence of $\esp{X_{< \md \Delta^2}} = \esp{X} - \esp{X_{\geq \md \Delta^2}}$ and $0 \leq \esp{X_{\geq \md \Delta^2}} \leq {\esp{X_{\geq \md \Delta^2}^2}}/({\md \Delta^2}) \leq \esp{X^2}/(\md \Delta^2) \leq (1/c_1) \cdot \esp{X}$ (using Fact \ref{fact:range}). The right hand side term is a direct consequence of the left one, and of the hypothesis $c_1 \cdot \esp{X} \leq \md \leq c_2 \cdot \esp{X}$.
\end{proof}

Our mean estimation algorithm works in two stages. We first compute a rough estimate $\md \in [2 \mus, 2500 \mus]$ with $\so{\ch \cdot \log(\hi/\lo)}$ \qss (where $0 < \lo < \mus < \hi$ are known bounds on $\mus$). Then, we improve the accuracy of the estimate to any value $\eps$, at extra cost $\so{\ch/\eps^{3/2}}$.


\boxalgo{Algo:basicEpsApprox}{$\eps-$approximation of the mean of a quantum sampler $\samp$.}{
{\bf Input:} a sampler $\samp$, an integer $\ch$, two values $0 < \lo < \hi$, two reals $0 < \eps, \delta < 1/2$. \\
{\bf Output:} an estimate $\muts$ of $\mus$.

\begin{enumerate}
  \item Set $\md = 8 \hi$ and $\pt = 0$
  \item While $\pt = 0$ and $\md \geq 2 \lo$:
    \begin{enumerate}
      \item Set $\md = \md/2$.
      \item Compute $\pt = \ampb{\samp, (0, \md \ch^2),25 \ch, \delta'}$ where $\delta' = \frac{\delta}{2 (3+\log(\hi/\lo))}$.
    \end{enumerate}
  \item If $\md < 2 \lo$ then \underline{output} $\muts = 0$.
  \item Else, compute $\qt = \ampb{\samp, (0, \eps^{-1} \md \ch^2),35^2 \eps^{-3/2} \ch, \delta/2}$ and \underline{output} $\muts = (\eps^{-1}\md \ch^2) \cdot \qt$.
\end{enumerate}
}

\begin{thm}
\label{Thm:basicEpsApprox}
  If $\ch \geq \phis / \mus$ and $\lo < \mus < \hi$ then the output $\muts$ of Algorithm \ref{Algo:basicEpsApprox} satisfies $|\muts - \mus| \leq \eps \mus$ with probability $1-\delta$. Moreover, for any $\ch, \lo, \hi$ it satisfies $\muts \leq (1+2\pi)^2 \mus$ with probability $1-\delta$. The number of \qss used by the algorithm is
  $\bo{\ch \cdot \left(\log\left(\frac{\hi}{\lo}\right)\log\left(\frac{\log(\hi/\lo)}{\delta}\right) + \eps^{-3/2}\log\left(\frac{1}{\delta}\right)\right)}$.
\end{thm}

\begin{proof}
  Assume that $\ch \geq \phis / \mus$ and $\lo < \mus < \hi$. We denote $p = (\md \ch^2)^{-1} \cdot \esp{v(\samp)_{< \md \ch^2}}$. By Lemma \ref{Lem:mdRange}, if $\md \geq 2500 \mus$ then $25\ch \leq \frac{1}{2 \sqrt{p}}$, and if $2\mus \leq \md \leq 4\mus$ then $25 \ch > \frac{8}{\sqrt{p}}$. Therefore, by Corollary \ref{Cor:BasicEst}, with probability $1-\delta'$, the value $\pt$ computed at Step 2.(b) is equal to $0$ when $\md \geq 2500 \mus$, and is different from $0$ when $2\mus \leq \md \leq 4\mus$. Thus, the first time Step 2.(b) of Algorithm \ref{Algo:basicEpsApprox} computes $\pt \neq 0$ happens for $\md \in [2 \mus, 2500 \mus]$, with probability at least  $(1-\delta')^{1+\log(4\hi/(2\mus))} > 1-\delta/2$.

  Consequently, we can assume that Step 4 is executed with $\md \in [2 \mus, 2500 \mus]$, and we let $\md' = \md/\eps$. According to Lemma \ref{Lem:mdRange} we have $(1-\eps/2) \mus \leq \esp{v(\samp)_{< \md' \ch^2}} \leq \mus$ and $35^2 \eps^{-3/2} \ch \geq \frac{8}{(\eps/2) \sqrt{q}}$, where $q = (\md' \ch^2)^{-1} \cdot \esp{v(\samp)_{< \md' \ch^2}}$. Thus, according to Corollary \ref{Cor:BasicEst}, the value $\qt$ satisfies $|\qt-q| \leq (\eps/2) q$ with probability $1-\delta/2$. Using the triangle inequality, it implies $|(\eps^{-1} \md \ch^2) \cdot \qt - \mus| \leq \eps \mus$.

  If $\lo \geq \mus$, this may only increase the probability to stop at Step 3 and output $\muts = 0$. If Step 4 is executed, we still have $\muts \leq (1+2\pi)^2 \mus$ with probability $1-\delta$, as a consequence of Corollary \ref{Cor:BasicEst}.
\end{proof}

\begin{rem}
  \label{Rem:RemoveL}
  If $\ch \geq \phis / \mus$ and $\hi > \mus$, observe that the output of Algorithm \ref{Algo:basicEpsApprox} satisfies $\muts = 0$ when $\lo \geq 1250 \mus$ and $\muts \neq 0$ when $\lo < \mus$, with probability $1-\delta$.
\end{rem}

We show in Appendix \ref{App:appEpsApprox} (Algorithm \ref{Algo:EpsApprox}) how to modify the last step of Algorithm \ref{Algo:basicEpsApprox} so that it uses $\so{\ch \cdot \eps^{-1} \log(1/\delta)}$ \qss only (Theorem \ref{Thm:EpsApprox}). Using Remark \ref{Rem:RemoveL}, we also remove the input parameter $\lo$ while keeping the number of \qss small in expectation (Algorithm \ref{Algo:EpsApproxMu}). Altogether, it leads to the following result.

\begin{thm}
\label{Thm:ExpectedEpsApprox}
  There is an algorithm that, given a sampler $\samp$, an integer $\ch$, a value $\hi > 0$, and two reals $0 < \eps, \delta < 1$, outputs an estimate $\muts$. If $\ch \geq \phis / \mus$ and $\hi > \mus$, it satisfies $|\muts - \mus| \leq \eps \mus$ with probability $1-\delta$, and the algorithm uses $\so{\ch \cdot \eps^{-1} \log^3(\hi/\mus)\log(1/\delta)}$ \qss in expectation.
\end{thm}

In Section \ref{Sec:optim}, we describe an $\Omega((\ch-1)/\eps)$ lower bound for this mean estimation problem. Before, we present three kinds of generalizations of the above algorithms.

\begin{itemize}
  \item \textbf{Higher moments.} Given an upper-bound $\ch^2 \geq (\esp{v(\samp)^k}/\esp{v(\samp)}^k)^{1/(k-1)}$ on the relative moment of order $k \geq 2$, one can easily generalize Facts \ref{fact:range}, Lemma \ref{Lem:mdRange} and Theorem \ref{Thm:EpsApprox} to show that $\mus$ can be estimated using $\so{\ch \cdot \eps^{-1/(2(k-1))} \log(\hi/\lo)\log(1/\delta)}$ \qss.

  \item \textbf{Implicit upper bound on $\phis/\mus$.} If instead of an explicit value $\ch \geq \phis/\mus$ we are given a non-increasing function $f$ such that $f(\mus) \geq \phis/\mus$, we can still estimate the mean $\mus$ using $\so{f(\mus/c) \cdot \eps^{-1} \log(\hi/\lo)\log(1/\delta)}$ \qss, where $c > 1$ is an absolute constant (Algorithm \ref{Algo:EpsApproxDec} in Appendix \ref{App:appDecreasingFunction}). The proof crucially uses the Markov-like inequality ``$\muts \leq (1+2\pi)^2 \mus$'' of Corollary \ref{Cor:BasicEst}.

  \item \textbf{Time complexity and variable-time samplers.} The \emph{time complexity} (number of quantum gates) of all above algorithms is essentially equal to the number of \qss multiplied by the time complexity $\tmax(\samp)$ of the considered sampler. Often, this last quantity is much larger than the more desirable \emph{$\ell_2$-average running time} $\talgs(\samp)$ defined by Ambainis \cite{Amb10b} in the context of \emph{variable-time amplitude amplification}. In Appendix~\ref{App:appVariableTime}, we develop a new \emph{variable-time amplitude estimation} algorithm (Theorem \ref{Thm:VarTime}), and we use it into our above algorithm to show that $\mus$ can be estimated \emph{in time} $\so{\ch \cdot \eps^{-2} \talgs(\samp) \cdot  \log^4(\tmax(\samp)) \log(\hi/\lo)\log(1/\delta)}$ (Theorem \ref{Thm:VarEpsApprox}).
\end{itemize}

The last two results are combined together in Section \ref{Sec:graphParam} and Appendix \ref{Sec:triangle} to describe an optimal quantum query algorithm that approximates the number of triangles in any graph.

\section{Optimality and separation results}
\label{Sec:optim}

Using a result due to Nayak and Wu \cite{NW99} on approximate counting, we can show a corresponding lower bound to Theorem \ref{Thm:ExpectedEpsApprox} already in the simple case of Bernoulli variables. For this purpose, we define that an algorithm $\alg$ solves the \emph{Mean Estimation problem for parameters $\eps,\Delta$} if, for any sampler $\samp$ satisfying $\phis / \mus \in [\Delta, 4\Delta]$ (the constant 4 is arbitrary), it outputs a value $\muts$ satisfying $|\muts-\mus| \leq \eps \mus$ with probability $2/3$.

\begin{thm}
\label{Thm:lowerCheb}
Any algorithm solving the Mean Estimation problem for parameters $0 < \eps < 1/5$ and $\Delta > 1$ on the sample space $\Omega = \{0,1\}$ must use $\om{(\Delta-1)/\eps}$ \qss.
\end{thm}

\begin{proof}
  Consider an algorithm $\alg$ solving the Mean Estimation problem for parameters $0 < \eps < 1/5$, $\Delta > 1$ using $N$ \qss. Take two integers $0 < t < n$ large enough such that $\sqrt{2} \Delta \leq \sqrt{n/t} \leq 4\Delta$ and $\eps t > 1$. For any oracle $\mathcal{O} : \{1,\dots,n\} \ra \rn$, define the quantum sampler $\samp_{\mathcal{O}}(\ket{0}\ket{0}) = \frac{1}{\sqrt{n}}\sum_{i \in [n]} \ket{i} \ket{\mathcal{O}(i)}$ and let $t_{\mathcal{O}} = |\{i \in [n] : \mathcal{O}(i) = 1\}|$. Observe that $\mu_{\samp_{\mathcal{O}}}=\phi_{\samp_{\mathcal{O}}}^2=t_{\mathcal{O}}/n$, and one \qs from $\samp_{\mathcal{O}}$ can be implemented with one quantum query to $\mathcal{O}$.

  According to \cite[Corollary 1.2]{NW99}, any algorithm that can distinguish $t_{\mathcal{O}} = t$ from $t_{\mathcal{O}} = \lceil(1+4\eps) t\rceil$ makes $\om{\sqrt{n/(\eps t)} + \sqrt{t(n-t)}/(\eps t)} = \om{(\sqrt{n/t}-1)/\eps} = \om{(\Delta-1)/\eps}$ quantum queries to $\mathcal{O}$. However, given the promise that $t_{\mathcal{O}} = t$ or $t_{\mathcal{O}} = \lceil(1+4\eps) t\rceil$ we can use $\alg$ with input $\samp_{\mathcal{O}}$, $\eps$, $\Delta$ to distinguish between the two cases using $N$ samples, that is $N$ queries to $\mathcal{O}$. Indeed, $\phi_{\samp_{\mathcal{O}}}/\mu_{\samp_{\mathcal{O}}} = \sqrt{n/t_{\mathcal{O}}} \in [\Delta,4\Delta]$ for such samplers (since $\lceil(1+4\eps) t\rceil \leq (1+5\eps) t \leq 2t$). Thus, $\alg$ must use $N = \om{(\Delta-1)/\eps}$ \qss.
\end{proof}


One may wonder whether the quantum speed-up presented in this paper holds if we only have access to copies of a quantum state $\sum_{x \in \Omega} \sqrt{d(x)} \ket{\psi_x} \ket{x}$ (instead of access to a unitary $\samp$ preparing it). Below we answer this question negatively. For this purpose, we define that an algorithm $\alg$ solves the \emph{state-based Mean Estimation problem for parameters $\eps,\Delta$} if, using access to some copies of an unknown state $\ket{d} = \sum_{x \in \Omega} \sqrt{d(x)}\ket{x}$ satisfying $\phi_d / \mu_d \in [\Delta, 4\Delta]$ (where $\mu_d = \sum_x d(x) x$ and $\phi_d^2 = \sum_x d(x) x^2$), it outputs a value $\mut_d$ satisfying $|\mut_d-\mu_d| \leq \eps \mu_d$ with probability $2/3$.

\begin{lem}
  \label{Lem:KLDiv}
  Consider two distributions $d, d'$ represented by the quantum states $\ket{d} = \sum_{x \in \Omega} \sqrt{d(x)}\ket{x}$ and $\ket{d'} = \sum_{x \in \Omega} \sqrt{d'(x)}\ket{x}$. The smallest integer $T$ needed to be able to discriminate $\ket{d}^{\otimes T}$ and $\ket{d'}^{\otimes T}$ with success probability $2/3$ satisfies $T \geq \frac{\ln(9/8)}{D(d || d')}$, where $D(d || d')$ is the KL-divergence from $d$ to $d'$.
\end{lem}

\begin{proof}
  According to Helstrom's bound \cite{Hel69} the best success probability to discriminate two states $\ket{\psi}$ and $\ket{\phi}$ is $\frac{1}{2}(1+\sqrt{1-|\ip{\psi}{\phi}|^2})$. Consequently, $T$ must satisfy $\frac{1}{2}(1+\sqrt{1-\ip{d}{d'}^{2T}}) \geq 2/3$, which implies
    \[T \geq \frac{\ln(9/8)}{-\ln(\ip{d}{d'}^2)}
      = \frac{\ln(9/8)}{-2\ln\left(\sum_x d(x) \sqrt{d'(x)/d(x)}\right)}
      \geq \frac{\ln(9/8)}{\sum_x d(x) \ln\left(d(x)/d'(x)\right)}
      = \frac{\ln(9/8)}{D(d || d')}\]
  where we used the concavity of the $-\ln$ function.
\end{proof}

\begin{thm}
  \label{Thm:lowerState}
  Any algorithm solving the state-based Mean Estimation problem for parameters $0 < \eps < 1/100$ and $\Delta > 1$ on the sample space $\Omega = \{0,1\}$ must use $\om{(\Delta^2-1)/\eps^2}$ copies of the input state.
\end{thm}

\begin{proof}
  Consider an algorithm $\alg$ solving the state-based Mean Estimation problem for parameters $0 < \eps < 1/100$, $\Delta > 1$ using $N$ copies of the input state. Given any $\ket{d} = \sqrt{1-p} \ket{0} + \sqrt{p} \ket{1}$ with $\phi_d / \mu_d \in [\sqrt{6} \Delta, \sqrt{8} \Delta]$ (notice that $\mu_d = \phi_d^2 = p$ and $1-p \geq 5/6 \geq 12 \eps$), we show how to construct a state $\ket{d'} = \sqrt{1-p'} \ket{0} + \sqrt{p'} \ket{1}$ such that
    \[\text{(1)} \ \ (1+4\eps) \mu_d < \mu_{d'} < (1+24\eps) \mu_d \ ; \qquad
      \text{(2)} \ \ \phi_{d'}/\mu_{d'} \in [\Delta, 4 \Delta] \ ; \qquad
      \text{(3)} \ \ D(d||d') \leq (12\eps)^2/(\Delta^2-1). \]
  It is clear that $\alg$ can be used to discriminate two such states. On the other hand, according to Lemma \ref{Lem:KLDiv}, any such algorithm muse use $N = \om{1/D(d||d')} = \om{(\Delta^2-1)/\eps^2}$ copies of the input state.

  The construction of $d'$ is adapted from \cite[Section 7]{DKLR00}. We set $p' = p e^{\alpha (1-p)}/\psi$ where $\alpha = 12\eps/(1-p) < 1$ and $\psi =(1-p) e^{-\alpha p} + p e^{\alpha (1-p)}$ (so that $1-p' = (1-p) e^{-\alpha p}/\psi$). We let $\dt{\psi}$ (resp. $\ddt{\psi}$) denote the first (resp. second) derivative of $\psi$ with respect to $\alpha$. A simple calculation shows that $\mu_{d'} - \mu_d = \dt{\psi}/\psi$ and $D(d||d') = \ln \psi$. Moreover,
    $\sigma_{d'}^2
      = \mathbb{E}_{x \sim d'} \left[(x - \mu_{d'})^2\right]
      = \mathbb{E}_{x \sim d'} \left[(x - \mu_{d})^2 \right] + 2 (\mu_d - \mu_{d'}) \mathbb{E}_{x \sim d'} \left[x - \mu_{d} \right] + (\mu_d - \mu_{d'})^2
      = \mathbb{E}_{x \sim d} \left[(x - p)^2 e^{\alpha (x-p) - \ln \psi} \right] - (\mu_d-\mu_{d'})^2
      = \ddt{\psi}/\psi - (\dt{\psi}/\psi)^2$.

  Since $\psi = \mathbb{E}_{x \sim d} \left[e^{\alpha (x - p)}\right]$, it can be deduced from the standard inequality $1 + u + u^2/3 \leq e^u \leq 1 + u + u^2$ (when $|u| \leq 1$) that $1 \leq 1 + \frac{p(1-p)}{3} \cdot \alpha^2 \leq \psi \leq 1 + p(1-p) \cdot \alpha^2 \leq 2$. Consequently, $\frac{2p(1-p)}{3} \cdot \alpha \leq \dt{\psi} \leq 2 p(1-p) \cdot \alpha$ and $\frac{2p(1-p)}{3} \leq \ddt{\psi} \leq 2 p(1-p)$. It implies that $4 \eps p \leq \mu_{d'} - \mu_d \leq 24 \eps p$ and $p(1-p)/3 - (24\eps p)^2 \leq \sigma_{d'}^2 \leq 2 p(1-p)$. Thus, $(1+4 \eps) \mu_d \leq \mu_{d'} \leq (1+24 \eps) \mu_d \leq \sqrt{2} \mu_d$ and $\frac{1}{6} \sigma_d^2/\mu_d^2 - (24\eps/\sqrt{2})^2 \leq \sigma_{d'}^2/\mu_{d'}^2 \leq 2 \sigma_d^2/\mu_d^2$. Since $\sigma_{d'}^2/\mu_{d'}^2 = \phi_{d'}^2/\mu_{d'}^2 - 1$ and $\phi_d / \mu_d \in [\sqrt{6} \Delta, \sqrt{8} \Delta]$, we obtain that $\Delta \leq \frac{1}{\sqrt{6}} \phi_{d} / \mu_{d} \leq \phi_{d'} / \mu_{d'} \leq \sqrt{2} \phi_{d} / \mu_{d} \leq 4 \Delta$. Finally, $D(d||d') = \ln \psi \leq p(1-p) \cdot \alpha^2 = (12\eps)^2 p / (1-p) \leq (12\eps)^2/(\Delta^2-1)$.
\end{proof}

\begin{rem}
  An intermediate version of Theorem \ref{Thm:lowerCheb} can be deduced from Theorem \ref{Thm:lowerState}, when $\samp$ is accessed via the \emph{reflection oracle} $\mathcal{O}_{\samp} = I - 2 \samp (\ket{0}\ket{0}) (\bra{0}\bra{0}) \samp^{-1}$ only (observe that this is the case for our algorithms). Indeed, according to \cite[Theorem 4]{JLS18}, for any algorithm performing $q$ queries to a reflection oracle $\mathcal{O} = I - 2 \proj{\phi}$, it is possible to remove the queries to $\mathcal{O}$ by using $\sim q^2$ copies of $\ket{\phi}$ instead.
\end{rem}

\section{Applications}
  We describe two applications of the Quantum Chebyshev Inequality. The first one (Section \ref{Sec:FreqMom}) concerns the computation of the frequency moments $F_k$ of order $k \geq 3$ in the streaming model. We design a $P$-pass algorithm with quantum memory $M$ satisfying a tradeoff of $P^2 M = \so{n^{1-2/k}}$, whereas the best algorithm with classical memory requires $P M = \Theta(n^{1-2/k})$. We then study (Section~\ref{Sec:graphParam}) the edge and triangle counting problems in the general graph model with quantum query access. We describe nearly optimal algorithms that approximate these parameters quadratically faster than in the classical query model.

  \subsection{Frequency moments in the multi-pass streaming model}
  \label{Sec:FreqMom}

In the streaming model with update (\emph{turnstile model}), the input is a vector $x \in \R^n$ obtained through a stream $\str = u_1, u_2, \dots$ of updates. Initially, $x(0) = (0,\dots,0)$, and each $u_j = (i,\lambda) \in [n] \times \R$ modifies the $i$-th coordinate of $x(j)$ by adding $\lambda$ to it. The goal of a \emph{streaming algorithm} $\astr$ is to output, at the end of the stream, some function of the final vector $x$ while minimizing the number $M \ll n$ of memory cells. In the \emph{multi-pass} model, the same stream is repeated for a certain number $P$ of passes, before the algorithm outputs its result.

The \emph{frequency moment of order $k$} is defined, for the final vector $x = (x_1,\dots,x_n)$, as $F_k(x) = \sum_{i \in [n]} |x_i|^k$. The problem of approximating $F_k$ when $k \geq 3$ has been addressed first with the AMS algorithm \cite{AMS99}, that uses $\bo{n^{1-1/k}}$ classical memory cells in the insertion-only model (where $u_j \in [n] \times \R^+$). A series of works in the turnstile model culminated in optimal one-pass algorithms with memory $\thet{n^{1-2/k}}$ \cite{LW13,Gan15}, and nearly optimal $P$-pass algorithms with memory $\sthet{n^{1-2/k}/P}$ \cite{MW10,AKO11,WZ12}. In the quantum setting, Montanaro \cite{Mon16} obtained a small improvement in terms of the approximation parameter $\eps$ only.


Our algorithm relies on a classical procedure for \emph{$\ell_2$ sampling}. Given $x \in \R^n$, we let $D_{q,x}$ denotes the \emph{$\ell_q$ distribution} that returns $i \in [n]$ with probability $\frac{|x_i|^q}{F_q(x)}$. One can observe that the (suboptimal) AMS algorithm \cite{AMS99} essentially samples $i \sim D_{1,x}$ and computes $F_1 \cdot |x_i|^{k-1}$. This is an unbiased estimator for $F_k(x)$ with variance $\bo{n^{1-1/k} F_k(x)^2}$ (thus requiring to compute $\bo{n^{1-1/k}}$ samples in one pass). Instead, we base our algorithm on the estimator $F_2(x) \cdot |x_i|^{k-2}$ where $i \sim D_{2,x}$. It reduces the variance to $\bo{n^{1-2/k} F_k(x)^2}$ \cite{MW10}, but it requires a procedure for $\ell_2$ sampling. To this end, we use the following algorithm from \cite{AKO11} to sample from an \emph{$(\eps,\delta)$-approximator} to $D_{2,x}$ (meaning that each $i \in [n]$ is sampled with a probability $p_i$ satisfying $(1-\eps) \frac{|x_i|^2}{F_2(x)} - \delta \leq p_i \leq (1+\eps) \frac{|x_i|^2}{F_2(x)} + \delta$).

\begin{thm}[\cite{AKO11}]
  \label{Thm:linSketchFk}
  There is a randomized streaming algorithm that, given a stream $\str$ with final vector $x$, a real $0 < \eps < 1/3$ and a value $\widetilde{F}_2$ such that $|\widetilde{F}_2 - F_2(x)| \leq (1/2) \cdot F_2(x)$, outputs a value $i \in [n]$ that is distributed according to an $(\eps,n^{-2})$-approximator to $D_{2,x}$. The algorithm uses $M = \bo{\eps^{-2} \log^3 n}$ classical memory cells. Moreover, each element of the stream is processed in time $T_{\mathit{upd}} = \bo{\eps^{-1} \log n}$, and the output is computed in time $T_{\mathit{rec}} = \bo{\eps^{-1} n \log n}$ after the last element is received.
\end{thm}

\boxest{Est:fkClass}{frequency moment $F_k$ of a stream.}{
{\bf Input:} a stream $\str$, an integer $k \geq 3$, a real $\widetilde{F}_2$, an approximation parameter $0 < \eps < 1$. \\
{\bf Output:} an estimate $\widetilde{F}_k$ of the frequency moment of order $k$ of $\str$.

\begin{enumerate}
  \item Compute $i \in [n]$ using the streaming algorithm of Theorem \ref{Thm:linSketchFk} with input $\str$, $\eps/4$, $\widetilde{F}_2$.
  \item Compute $x_i$ using a second pass over $\str$.
  \item \underline{Output} $\widetilde{F}_2 \cdot |x_i|^{k-2}$.
\end{enumerate}
}

\begin{prop}[\cite{MW10,AKO11}]
  \label{Prop:fkClass}
  If we let $X$ denote the output random variable of Estimator \ref{Est:fkClass}, then $\esp{X} = (1 \pm \eps/2) F_k$ and $\var{X} \leq \bo{n^{1-2/k} F_k^2}$, when $|\widetilde{F}_2 - F_2| \leq (\eps/4) \cdot F_2$.
\end{prop}

Using standard techniques, the algorithm of Estimator \ref{Est:fkClass} can be made reversible and therefore implemented by a quantum sampler $\samp$. We need to be careful that the reverse computation $\samp^{-1}$ can also be done efficiently. Usually, that would require processing the same stream but in the \emph{reverse} direction. However, the construction given in \cite{AKO11} has the particularity to be a \emph{linear sketch} algorithm (the memory content is a linear function $L(x)$ of the input $x$, see Definition \ref{Def:linSketch}). In Appendix~\ref{App:appFrequencyMoments} (Proposition \ref{Prop:reversibleSketch}), we show that the reverse computation of such algorithms can be done efficiently with one pass in the \emph{direct} direction. We combine the quantum sampler that is obtained from this result with the Quantum Chebyshev Inequality (Theorem \ref{Thm:ExpectedEpsApprox}) to obtain the following tradeoff.

\begin{thm}
  \label{Thm:fkQuant}
  There is a quantum streaming algorithm that, given a stream $\str$, two integers $P \geq 1$, $k \geq 3$ and an approximation parameter $0 < \eps < 1$, outputs an estimate $\widetilde{F}_k$ such that $|\widetilde{F}_k - F_k| \leq \eps F_k$ with probability $2/3$. The algorithm uses $\so{n^{1-2/k}/(\eps P)^2}$ quantum memory cells, and it makes $\so{P \cdot (k \log n + \eps^{-1})}$ passes over the stream $\str$.
\end{thm}

\begin{proof}
  We first compute, in one pass, a value $\widetilde{F}_2$ such that $|\widetilde{F}_2 - F_2| \leq (\eps/2) F_2$ with high probability, using \cite{AMS99,Mon16} for instance. The complexity is absorbed by the final result.
  Then, using Estimator \ref{Est:fkClass} together with Proposition \ref{Prop:reversibleSketch}, we can design a quantum sampler $\samp$ using memory $M = \so{\eps^{-2} \log^3 n}$ such that $\samp (\ket{0}\ket{0}) = \sum_{r \in \rn^M} \ket{r}\ket{\psi_r} \ket{f_r}$ where each $\ket{r}$ corresponds to a different random seed for the linear sketch algorithm of Theorem \ref{Thm:linSketchFk}, $\ket{f_r}$ is the output of Estimator \ref{Est:fkClass}, and $\ket{\psi_r}$ is some garbage state obtained when making Estimator \ref{Est:fkClass} reversible. According to Proposition \ref{Prop:fkClass}, we have $\mus = (1 \pm \eps/2) F_k$ and $\sigs \leq \bo{\sqrt{n^{1-2/k}} F_k}$. Moreover one \qs can be implemented with two passes over the stream.

  We concatenate $Q = n^{1-2/k}/P^2$ such samplers, and compute the mean $\bar{f} = Q^{-1} \cdot (f_{r_1} + \dots + f_{r_Q})$ of their results, i.e. $\bar{\samp} (\ket{0}\ket{0}) = \sum_{r_1,\dots,r_Q \in \rn^M} \ket{r_1,\dots,r_Q} \ket{\psi_1,\dots,\psi_Q} \ket{f_{r_1}, \dots, f_{r_Q}} \ket{\bar{f}}$. This sampler satisfies $\sigma_{\bar{\samp}} \leq \bo{P F_k}$, and it requires two passes and memory $\bar{M} =  \so{Q \cdot \eps^{-2} \log^3 n}$ to be implemented.
  Finally, we approximate $F_k$ by applying Theorem \ref{Thm:ExpectedEpsApprox} on $\bar{\samp}$, which uses $\so{P \cdot (k \log n + \eps^{-1})}$ \qss.
\end{proof}


  \subsection{Approximating graph parameters in the query model}
  \label{Sec:graphParam}

In this section, we consider the \emph{general graph model} \cite{KKR04,Gol17} that provides query access to a graph $G = (V,E)$ through the following operations: (1) \emph{degree query} (given $v \in V$, returns the degree $d_v$ of $v$), (2) \emph{neighbor query} (given $v \in V$ and $i$, returns the $i$-th neighbor of $v$ if $i \leq d_v$, and $\bot$ otherwise), and (3) \emph{vertex-pair query} (given $u,v \in V$, indicates if $(u,v) \in E$). This is a combination of the dense graph model (pair queries) and the bounded-degree model (neighbor and degree queries). We refer the reader to \cite[Chapter~10]{Gol17} for a more detailed discussion about it. It can be extended to the standard quantum query framework. A quantum degree query is represented as a unitary $\mathcal{O}_{deg}$ such that $\mathcal{O}_{deg} \ket{v} \ket{b} = \ket{v} \ket{y \oplus d_v}$ where $v \in V$ and $b \in \rn^{\lceil\log n\rceil}$. The quantum neighbor $\mathcal{O}_{neigh}$ and vertex-pair $\mathcal{O}_{pair}$ queries are defined similarly. The \emph{query complexity} of an algorithm in the \emph{quantum general graph model} is the number of times it uses $\mathcal{O}_{deg}$, $\mathcal{O}_{nei}$ or $\mathcal{O}_{pair}$.

In the following, we let $n$ denote the number of vertices, $m$ the number of edges and $t$ the number of triangles in $G$.
We consider the problems of estimating $m$ and $t$, for which we provide nearly optimal quantum algorithms. The description and analysis of these algorithms is deferred to Appendix~\ref{App:appGraphParameters}.

\paragraph{Edge counting}
In the classical setting, with degree queries only, Feige \cite{Fei06} showed that $\thet{n/(\eps \sqrt{m})}$ queries are sufficient to compute a factor $(2+\eps)$ approximation of $m$, but no factor $(2-\eps)$ approximation can be obtained in sublinear time. Using both degree and neighbor queries, it is possible to compute a factor $(1+\eps)$ approximation with $\thet{n/(\sqrt{\eps m})}$ classical queries \cite{GR08,Ses15,ERS17a}. These results were generalized to $k$-star counting in \cite{GRS11,ERS17a}. In the quantum setting, we prove the following results in Appendix \ref{Sec:edge}.

\begin{thm}
  \label{Thm:edge}
  There is an algorithm that, given query access to any $n$-vertex graph $G$ with $m$ edges, and an approximation parameter $\eps < 1$, outputs an estimate $\mt$ of $m$ such that $|\mt-m| \leq \eps m$ with probability $2/3$. This algorithm performs $\so{\frac{n^{1/2}}{\eps m^{1/4}}}$ quantum degree and neighbor queries in expectation. Moreover, it does not use vertex-pair queries.
\end{thm}

\begin{thm}
  \label{Thm:edgeLB}
  Any algorithm that computes an $\eps$-approximation of the number $m$ of edges in any $n$-vertex graph, given query access to it, must use $\om{\frac{n^{1/2}}{(\eps m)^{1/4}} \cdot \log^{-1}(n)}$ quantum queries in expectation.
\end{thm}

\paragraph{Triangle counting}
In the classical general graph model, the triangle counting problem requires $\widetilde{\Theta}(n/t^{1/3}+\min(m,m^{3/2}/t))$ queries in expectation \cite{ELR15,ELRS17}. This result was generalized to $k$-clique counting in \cite{ERS18}. In the quantum setting, we prove the following results in Appendix \ref{Sec:triangle}.

\begin{thm}
  \label{Thm:triangleEps}
  There is an algorithm that, given query access to any $n$-vertex graph $G$ with $m$ edges and $t$ triangles, and an approximation parameter $\eps < 1$, outputs an estimate $\tti$ of $t$ such that $|\tti-t| \leq \eps t$ with probability $2/3$. This algorithm performs $\so{\left(\frac{\sqrt{n}}{t^{1/6}} + \frac{m^{3/4}}{\sqrt{t}}\right) \cdot \poly(1/\eps)}$ quantum queries in expectation.
\end{thm}

\begin{thm}
  \label{Thm:triangleLB}
  Any algorithm that computes an $\eps$-approximation to the number $t$ of triangles in any $n$-vertex graph with $m$ vertices, given query access to it, must use $\om{\left(\frac{\sqrt{n}}{t^{1/6}} + \frac{m^{3/4}}{\sqrt{t}}\right) \cdot \log^{-1}(n)}$ quantum queries in expectation.
\end{thm}

\section{Open questions}
  Is it possible to improve the complexity of our main result (Theorem \ref{Thm:ExpectedEpsApprox}) to $\bo{\ch/\eps}$ exactly? Can we generalize it to sample spaces with negative values? What are other possible applications? Two promising problems are minimum spanning tree weight \cite{CRT05} and arbitrary subgraph counting \cite{ERS18,AKK18}.

\section*{Acknowledgements}
  The authors want to thank the anonymous referees for their valuable comments and suggestions which helped to improve this paper.


{\small
\bibliographystyle{plain} 
\bibliography{Bibliography}

\begin{thebibliography}{10}

\bibitem{AA05}
S.~Aaronson and A.~Ambainis.
\newblock Quantum search of spatial regions.
\newblock {\em Theory of Computing}, 1(4):47--79, 2005.

\bibitem{AT07}
D.~Aharonov and A.~Ta-Shma.
\newblock Adiabatic quantum state generation.
\newblock {\em SIAM Journal on Computing}, 37(1):47--82, 2007.

\bibitem{AHLW16}
Y.~Ai, W.~Hu, Y.~Li, and D.~P. Woodruff.
\newblock New characterizations in turnstile streams with applications.
\newblock In {\em Proceedings of the 31st Conference on Computational
  Complexity}, CCC '16, pages 20:1--20:22, 2016.

\bibitem{AMS99}
N.~Alon, Y.~Matias, and M.~Szegedy.
\newblock The space complexity of approximating the frequency moments.
\newblock {\em J. Comput. Syst. Sci.}, 58(1):137--147, 1999.

\bibitem{Amb10b}
A.~Ambainis.
\newblock Variable time amplitude amplification and a faster quantum algorithm
  for solving systems of linear equations.
\newblock Technical Report arxiv:1010.4458, arXiv.org, 2010.

\bibitem{AKO11}
A.~Andoni, R.~Krauthgamer, and K.~Onak.
\newblock Streaming algorithms via precision sampling.
\newblock In {\em Proceedings of the 52nd Symposium on Foundations of Computer
  Science}, FOCS '11, pages 363--372, 2011.

\bibitem{AW17}
S.~Arunachalam and R.~de~Wolf.
\newblock Optimal quantum sample complexity of learning algorithms.
\newblock In {\em Proceedings of the 32nd Computational Complexity Conference},
  CCC '17, pages 25:1--25:31, 2017.

\bibitem{AKK18}
S.~Assadi, M.~Kapralov, and S.~Khanna.
\newblock A simple sublinear-time algorithm for counting arbitrary subgraphs
  via edge sampling.
\newblock Technical Report arxiv:1811.07780, arXiv.org, 2018.

\bibitem{BOW17}
C.~Badescu, R.~O'Donnell, and J.~Wright.
\newblock Quantum state certification.
\newblock Technical Report arxiv:1708.06002, arXiv.org, 2017.

\bibitem{BFRSW13}
T.~Batu, L.~Fortnow, R.~Rubinfeld, W.~D. Smith, and P.~White.
\newblock Testing closeness of discrete distributions.
\newblock {\em J. ACM}, 60(1):4:1--4:25, 2013.

\bibitem{Ben89}
C.~Bennett.
\newblock Time/space trade-offs for reversible computation.
\newblock {\em SIAM Journal on Computing}, 18(4):766--776, 1989.

\bibitem{BDGT11}
G.~Brassard, F.~Dupuis, S.~Gambs, and A.~Tapp.
\newblock An optimal quantum algorithm to approximate the mean and its
  application for approximating the median of a set of points over an arbitrary
  distance.
\newblock Technical Report arxiv:1106.4267, arXiv.org, 2011.

\bibitem{BHMT02}
G.~Brassard, P.~H{\o}yer, M.~Mosca, and A.~Tapp.
\newblock Quantum amplitude amplification and estimation.
\newblock {\em Quantum Computation and Quantum Information: A Millennium
  Volume}, 1:53--74, 2002.

\bibitem{BHH11}
S.~Bravyi, A.~W. Harrow, and A.~Hassidim.
\newblock Quantum algorithms for testing properties of distributions.
\newblock {\em IEEE Transactions on Information Theory}, 57(6):3971--3981,
  2011.

\bibitem{BCW98}
H.~Buhrman, R.~Cleve, and A.~Wigderson.
\newblock Quantum vs. classical communication and computation.
\newblock In {\em Proceedings of the 33th Symposium on Theory of Computing},
  STOC '98, pages 63--68, 1998.

\bibitem{CDKS18}
C.~L. Canonne, I.~Diakonikolas, D.~M. Kane, and A.~Stewart.
\newblock Testing conditional independence of discrete distributions.
\newblock In {\em Proceedings of the 50th Symposium on Theory of Computing},
  STOC '18, pages 735--748, 2018.

\bibitem{CCKM13}
A.~Chakrabarti, G.~Cormode, R.~Kondapally, and A.~McGregor.
\newblock Information cost tradeoffs for augmented index and streaming language
  recognition.
\newblock {\em SIAM Journal on Computing}, 42(1):61--83, 2013.

\bibitem{CGJ18}
S.~Chakraborty, A.~Gily{\'{e}}n, and S.~Jeffery.
\newblock The power of block-encoded matrix powers: improved regression
  techniques via faster {Hamiltonian} simulation.
\newblock Technical Report arxiv:1804.01973, arXiv.org, 2018.

\bibitem{CDVV14}
S.~Chan, I.~Diakonikolas, P.~Valiant, and G.~Valiant.
\newblock Optimal algorithms for testing closeness of discrete distributions.
\newblock In {\em Proceedings of the 25th Symposium on Discrete Algorithms},
  SODA '14, pages 1193--1203, 2014.

\bibitem{CRT05}
B.~Chazelle, R.~Rubinfeld, and L.~Trevisan.
\newblock Approximating the minimum spanning tree weight in sublinear time.
\newblock {\em SIAM Journal on Computing}, 34(6):1370--1379, 2005.

\bibitem{CS17}
A.~N. Chowdhury and R.~D. Somma.
\newblock Quantum algorithms for {Gibbs} sampling and hitting-time estimation.
\newblock {\em Quantum Info. Comput.}, 17(1-2):41--64, 2017.

\bibitem{DKLR00}
P.~Dagum, R.~Karp, M.~Luby, and S.~Ross.
\newblock An optimal algorithm for {Monte Carlo} estimation.
\newblock {\em SIAM Journal on Computing}, 29(5):1484--1496, 2000.

\bibitem{DGG10}
N.~Destainville, B.~Georgeot, and O.~Giraud.
\newblock Quantum algorithm for exact {Monte Carlo} sampling.
\newblock {\em Phys. Rev. Lett.}, 104:250502, 2010.

\bibitem{DFK91}
M.~Dyer, A.~Frieze, and R.~Kannan.
\newblock A random polynomial-time algorithm for approximating the volume of
  convex bodies.
\newblock {\em J. ACM}, 38(1):1--17, 1991.

\bibitem{ELR15}
T.~Eden, A.~Levi, and D.~Ron.
\newblock Approximately counting triangles in sublinear time.
\newblock Technical Report TR15-046, ECCC, 2015.

\bibitem{ELRS17}
T.~Eden, A.~Levi, D.~Ron, and C.~Seshadhri.
\newblock Approximately counting triangles in sublinear time.
\newblock {\em {SIAM} J. Comput.}, 46(5):1603--1646, 2017.

\bibitem{ERS17a}
T.~Eden, D.~Ron, and C.~Seshadhri.
\newblock Sublinear time estimation of degree distribution moments: The
  degeneracy connection.
\newblock In {\em Proceedings of the 44th International Colloquium on Automata,
  Languages, and Programming}, ICALP '17, pages 7:1--7:13, 2017.

\bibitem{ERS18}
T.~Eden, D.~Ron, and C.~Seshadhri.
\newblock On approximating the number of k-cliques in sublinear time.
\newblock In {\em Proceedings of the 50th Symposium on Theory of Computing},
  STOC '18, pages 722--734, 2018.

\bibitem{ER18}
T.~Eden and W.~Rosenbaum.
\newblock Lower bounds for approximating graph parameters via communication
  complexity.
\newblock In {\em Proceedings of the Workshop on Approximation, Randomization,
  and Combinatorial Optimization: Algorithms and Techniques}, APPROX/RANDOM
  '18, pages 11:1--11:18, 2018.

\bibitem{Fei06}
U.~Feige.
\newblock On sums of independent random variables with unbounded variance and
  estimating the average degree in a graph.
\newblock {\em SIAM Journal on Computing}, 35(4):964--984, 2006.

\bibitem{Gan15}
S.~Ganguly.
\newblock Taylor polynomial estimator for estimating frequency moments.
\newblock In {\em Proceedings of the 42nd International Colloquium on Automata,
  Languages and Programming}, ICALP '15, pages 542--553, 2015.

\bibitem{Gol17}
O.~Goldreich.
\newblock {\em Introduction to Property Testing}.
\newblock Cambridge University Press, 2017.

\bibitem{GR08}
O.~Goldreich and D.~Ron.
\newblock Approximating average parameters of graphs.
\newblock {\em Random Struct. Algorithms}, 32(4):473--493, 2008.

\bibitem{GR11}
O.~Goldreich and D.~Ron.
\newblock On testing expansion in bounded-degree graphs.
\newblock In {\em Studies in Complexity and Cryptography. Miscellanea on the
  Interplay between Randomness and Computation}, pages 68--75. Springer-Verlag,
  2011.

\bibitem{GRS11}
M.~Gonen, D.~Ron, and Y.~Shavitt.
\newblock Counting stars and other small subgraphs in sublinear-time.
\newblock {\em SIAM Journal on Discrete Mathematics}, 25(3):1365--1411, 2011.

\bibitem{Hei02}
S.~Heinrich.
\newblock Quantum summation with an application to integration.
\newblock {\em Journal of Complexity}, 18(1):1 -- 50, 2002.

\bibitem{Hel69}
C.~W. Helstrom.
\newblock Quantum detection and estimation theory.
\newblock {\em Journal of Statistical Physics}, 1(2):231--252, Jun 1969.

\bibitem{JN14}
R.~Jain and A.~Nayak.
\newblock The space complexity of recognizing well-parenthesized expressions in
  the streaming model: the index function revisited.
\newblock {\em IEEE Transactions on Information Theory}, 60(10):6646--6668,
  2014.

\bibitem{JS96}
M.~Jerrum and A.~Sinclair.
\newblock The {Markov} chain {Monte Carlo} method: An approach to approximate
  counting and integration.
\newblock In {\em Approximation Algorithms for {NP}-hard Problems}, chapter~12,
  pages 482--520. PWS Publishing, 1996.

\bibitem{JSV04}
M.~Jerrum, A.~Sinclair, and E.~Vigoda.
\newblock A polynomial-time approximation algorithm for the permanent of a
  matrix with nonnegative entries.
\newblock {\em J. ACM}, 51(4):671--697, 2004.

\bibitem{JVV86}
M.~R. Jerrum, L.~G. Valiant, and V.~V. Vazirani.
\newblock Random generation of combinatorial structures from a uniform
  distribution.
\newblock {\em Theoretical Computer Science}, 43:169 -- 188, 1986.

\bibitem{JLS18}
Z.~Ji, Y.-K. Liu, and F.~Song.
\newblock Pseudorandom quantum states.
\newblock In {\em Advances in Cryptology}, CRYPTO '18, pages 126--152, 2018.

\bibitem{KL83}
R.~M. Karp and M.~Luby.
\newblock {Monte-Carlo} algorithms for enumeration and reliability problems.
\newblock In {\em Proceedings of the 24th Symposium on Foundations of Computer
  Science}, FOCS '83, pages 56--64, 1983.

\bibitem{KKR04}
T.~Kaufman, M.~Krivelevich, and D.~Ron.
\newblock Tight bounds for testing bipartiteness in general graphs.
\newblock {\em SIAM Journal on Computing}, 33(6):1441--1483, 2004.

\bibitem{KOS07}
E.~Knill, G.~Ortiz, and R.~D. Somma.
\newblock Optimal quantum measurements of expectation values of observables.
\newblock {\em Phys. Rev. A}, 75:012328, 2007.

\bibitem{Gal09}
F.~Le~Gall.
\newblock Exponential separation of quantum and classical online space
  complexity.
\newblock {\em Theor. Comp. Sys.}, 45(2):188--202, 2009.

\bibitem{LW17}
T.~Li and X.~Wu.
\newblock Quantum query complexity of entropy estimation.
\newblock Technical Report arxiv:1710.06025, arXiv.org, 2017.

\bibitem{LNW14}
Y.~Li, H.~L. Nguyen, and D.~P. Woodruff.
\newblock Turnstile streaming algorithms might as well be linear sketches.
\newblock In {\em Proceedings of the 46th Symposium on Theory of Computing},
  STOC '14, pages 174--183, 2014.

\bibitem{LW13}
Y.~Li and D.~P. Woodruff.
\newblock A tight lower bound for high frequency moment estimation with small
  error.
\newblock In {\em Proceedings of the Workshop on Approximation, Randomization,
  and Combinatorial Optimization: Algorithms and Techniques}, APPROX/RANDOM
  '13, pages 623--638, 2013.

\bibitem{MMN14}
F.~Magniez, C.~Mathieu, and A.~Nayak.
\newblock Recognizing well-parenthesized expressions in the streaming model.
\newblock {\em SIAM Journal on Computing}, 43(6):1880--1905, 2014.

\bibitem{MW10}
M.~Monemizadeh and D.~P. Woodruff.
\newblock 1-pass relative-error {Lp-sampling} with applications.
\newblock In {\em Proceedings of the 21st Symposium on Discrete Algorithms},
  SODA '10, pages 1143--1160, 2010.

\bibitem{Mon15}
A.~Montanaro.
\newblock Quantum speedup of {Monte Carlo} methods.
\newblock {\em Proceedings of the Royal Society of London A: Mathematical,
  Physical and Engineering Sciences}, 471(2181), 2015.

\bibitem{Mon16}
A.~Montanaro.
\newblock The quantum complexity of approximating the frequency moments.
\newblock {\em Quantum Information and Computation}, 16:1169--1190, 2016.

\bibitem{NT17}
A.~Nayak and D.~Touchette.
\newblock Augmented index and quantum streaming algorithms for {DYCK(2)}.
\newblock In {\em Proceedings of the 32nd Conference on Computational
  Complexity}, CCC '17, pages 23:1--23:21, 2017.

\bibitem{NW99}
A.~Nayak and F.~Wu.
\newblock The quantum query complexity of approximating the median and related
  statistics.
\newblock In {\em Proceedings of the 31st Symposium on Theory of Computing},
  STOC '99, pages 384--393, 1999.

\bibitem{PW09}
D.~Poulin and P.~Wocjan.
\newblock Sampling from the thermal quantum {Gibbs} state and evaluating
  partition functions with a quantum computer.
\newblock {\em Phys. Rev. Lett.}, 103:220502, 2009.

\bibitem{Raz03}
A.~A. Razborov.
\newblock Quantum communication complexity of symmetric predicates.
\newblock {\em Izvestiya: Mathematics}, 67(1):145--159, 2003.

\bibitem{Ses15}
C.~Seshadhri.
\newblock A simpler sublinear algorithm for approximating the triangle count.
\newblock Technical Report arxiv:1505.01927, arXiv.org, 2015.

\bibitem{TOVP11}
K.~Temme, T.~J. Osborne, K.~Vollbrecht, D.~Poulin, and F.~Verstraete.
\newblock Quantum metropolis sampling.
\newblock {\em Nature}, 471:87, 2011.

\bibitem{SVV09}
D.~\v{S}tefankovi\v{c}, S.~Vempala, and E.~Vigoda.
\newblock Adaptive simulated annealing: A near-optimal connection between
  sampling and counting.
\newblock {\em J. ACM}, 56(3):18:1--18:36, 2009.

\bibitem{WA08}
P.~Wocjan and A.~Abeyesinghe.
\newblock Speedup via quantum sampling.
\newblock {\em Phys. Rev. A}, 78:042336, 2008.

\bibitem{WCNA09}
P.~Wocjan, C.-F. Chiang, D.~Nagaj, and A.~Abeyesinghe.
\newblock Quantum algorithm for approximating partition functions.
\newblock {\em Phys. Rev. A}, 80:022340, 2009.

\bibitem{WZ12}
D.~P. Woodruff and Q.~Zhang.
\newblock Tight bounds for distributed functional monitoring.
\newblock In {\em Proceedings of the 44th Symposium on Theory of Computing},
  STOC '12, pages 941--960, 2012.

\end{thebibliography}
}

\appendix

  \section{A faster algorithm for mean approximation}
  \label{App:appEpsApprox}

We show first how to improve the dependence on $\eps$ of Algorithm \ref{Algo:basicEpsApprox}. To this end, we need a finer version of an algorithm from \cite{Hei02,Mon15}, where we introduce a new parameter $\Gamma$ (the result presented in \cite{Mon15} corresponds to $\Gamma = 1$).

\boxalgo{Algo:SApprox}{subroutine for approximating the mean of a quantum sampler $\samp$.}{
{\bf Input:} a sampler $\samp$, a parameter $\wid > 0$, an integer $t > 2$, a failure parameter $0 < \delta < 1$. \\
{\bf Output:} an estimate $\muts$ of $\mus$.

\begin{enumerate}
  \item Set $k = \lceil\log t \rceil - 1$, $t_0 = \left\lceil 3\pi^2 t \sqrt{\log t}\right\rceil$.
  \item Compute $\pt_0 = \ampb{\samp,(0,\wid),t_0,\delta/(k+1)}$.
  \item For $\ell=1,\dots, k$:
    \begin{enumerate}
    \item Compute $\pt_{\ell} = \ampb{\samp,(2^{\ell-1}\wid,2^{\ell}\wid),t_0,\delta/(k+1)}$.
    \end{enumerate}
  \item \underline{Output} $\muts = \sum_{\ell=0}^k 2^{\ell} \wid \cdot \pt_\ell$.
\end{enumerate}
}

\begin{prop}
\label{Prop:SApprox}
The output $\muts$ of Algorithm \ref{Algo:SApprox} satisfies $ |\muts - \mus| \leq \frac{1}{t} \left(\sqrt{\wid} + \frac{\phis}{\sqrt{\wid}}\right)^2$ and $\muts \leq (1+2\pi)^2 \mus$ with probability $1-\delta$. The number of \qss used by the algorithm is $\mathcal{O}(t \log^{3/2}(t) \log(\log(t)/\delta)$.
\end{prop}

\begin{proof}
  Observe that $\mus = \sum_{\ell = 0}^k 2^{\ell}\wid \cdot p_{\ell} + \esp{v(\samp)_{\geq 2^{k}\wid}}$, where $p_{0} = \frac{\esp{v(\samp)_{0,\wid}}}{\wid}$ and $p_{\ell} = \frac{\esp{v(\samp)_{2^{\ell-1}\wid,2^{\ell}\wid}}}{2^{\ell}\wid}$. Using Corollary \ref{Cor:BasicEst} and a union bound, we can assume  $|\pt_{\ell} - p_{\ell}| \leq \pi^2 \left(\frac{\sqrt{p_{\ell}}}{t_0} + \frac{1}{t_0^2}\right)$ and $\pt_{\ell} \leq (1+2\pi)^2 p_{\ell}$ for all $\ell$, with probability $1-\delta$. It implies $\muts \leq (1+2\pi)^2 \muts$. On the other hand, using the triangle inequality,
    \begin{align*}
      |\muts - \mus|
          & \leq \pi^2 \left(\frac{\wid}{t_0} + \frac{1}{t_0} \sum_{\ell = 1}^k \sqrt{2^{\ell} \wid \cdot \esp{v(\samp)_{2^{\ell-1}\wid,2^{\ell}\wid}}} + \frac{\wid}{t_0^2} \sum_{\ell = 0}^k 2^{\ell}\right) + \esp{v(\samp)_{\geq 2^{k}\wid}} \\
          &  \leq \pi^2 \left(\frac{\wid}{t_0} + \frac{1}{t_0} \sqrt{k} \sqrt{\sum_{\ell = 1}^k 2^{\ell}\wid \cdot \frac{\esp{v(\samp)_{2^{\ell-1} \wid,2^{\ell} \wid}^2}}{2^{\ell-1} \wid}} + \frac{2^{k+1}}{t_0^2} \wid\right) + \frac{\phis^2}{2^k \wid} \\
          &  \leq \pi^2 \left(\frac{\wid}{t_0} + \frac{\sqrt{2k}}{t_0} \cdot \phis + \frac{2^{k+1}}{t_0^2} \wid\right) + \frac{\phis^2}{2^k \wid}
          \leq \frac{1}{t} \left(\sqrt{\wid} + \frac{\phis}{\sqrt{\wid}}\right)^2
    \end{align*}
  where we used Fact \ref{fact:range} and the Cauchy-Schwarz inequality, at the second step.
\end{proof}

If we set $\wid = \phis$ in the above inequality, we obtain $|\muts - \mus| \leq 4\phis/t$, and thus $|\muts - \mus| \leq \eps \mus$ when $t = \om{\eps^{-1} \ch}$. Since $\phis$ is unknown, we approximate it by $\phits = \md \ch$ instead, where $\md \in [2 \mus, 2500 \mus]$ is obtained with the same method as in Algorithm \ref{Algo:basicEpsApprox}.

\boxalgo{Algo:EpsApprox}{$\eps-$approximation of the mean of a quantum sampler $\samp$.}{
{\bf Input:} a sampler $\samp$, an integer $\ch$, two values $0 < \lo < \hi$, two reals $0 < \eps, \delta < 1/2$. \\
{\bf Output:} an estimate $\muts$ of $\mus$.

\begin{enumerate}
  \item Set $\md = 8 \hi$ and $\pt = 0$
  \item While $\pt = 0$ and $\md \geq 2 \lo$:
    \begin{enumerate}
      \item Set $\md = \md/2$.
      \item Compute $\pt = \ampb{\samp, (0, \md \ch^2),25 \ch, \delta'}$ where $\delta' = \frac{\delta}{2 (3+\log(\hi/\lo))}$.
    \end{enumerate}
  \item If $\md < 2 \lo$ then \underline{output} $\muts = 0$.
  \item Else, run Algorithm \ref{Algo:SApprox} on input $\samp$, $\wid = \md \cdot \ch$, $t = 51^2 \eps^{-1} \ch$, $\delta/2$ and \underline{output} the result as $\muts$.
\end{enumerate}
}

\begin{thm}
\label{Thm:EpsApprox}
  If $\ch \geq \phis / \mus$ and $\lo < \mus < \hi$ then the output $\muts$ of Algorithm \ref{Algo:EpsApprox} satisfies $|\muts - \mus| \leq \eps \mus$ with probability $1-\delta$. Moreover, for any $\ch, \lo, \hi$ it satisfies $\muts \leq (1+2\pi)^2 \mus$ with probability $1-\delta$. The number of \qss used by the algorithm is
    \[\bo{\ch \cdot \left(\log\left(\frac{\hi}{\lo}\right)\log\left(\frac{\log(\hi/\lo)}{\delta}\right) + \eps^{-1} \log^{3/2}(\ch) \log\left(\frac{\log \ch}{\delta}\right)\right)}.\]
\end{thm}

\begin{proof}
  Steps 1 to 3 are identical to the beginning of Algorithm \ref{Algo:basicEpsApprox}. Consequently, by the same arguments as in the proof of Theorem \ref{Thm:basicEpsApprox}, when $\ch \geq \phis / \mus$ and $\lo < \mus < \hi$ we can assume (with probability $1-\delta/2$) that Step 4 is executed with $\md \in [2 \mus, 2500 \mus]$. In this case, according to Proposition \ref{Prop:SApprox}, the output $\muts$ satisfies $|\muts-\mus| \leq \frac{1}{51^2 \eps^{-1} \ch} \left(\sqrt{2500 \mus \ch} + \frac{\phis}{\sqrt{2 \mus \ch}} \right)^2 \leq \frac{(\sqrt{2500} + 1/\sqrt{2})^2}{51^2} \eps \mus \leq \eps \mus$ with probability $1-\delta/2$.
\end{proof}

The next algorithm details how to replace the input parameter $\lo$ with a logarithmic search on decreasing values of $\lo$. This causes the factor $\log(\hi/\lo)$ in the complexity bounds to become $\log^3(\hi/\mus)$. A similar result can be obtained for all the other algorithms of Section \ref{Sec:cheb}.

\boxalgo{Algo:EpsApproxMu}{$\eps-$approximation of the mean of a quantum sampler $\samp$.}{
{\bf Input:} a sampler $\samp$, an integer $\ch$, a value $\hi > 0$, two reals $0 < \eps, \delta < 1/2$. \\
{\bf Output:} an estimate $\muts$ of $\mus$.

\begin{enumerate}
  \item Set $i = 1$.
  \item Run Algorithm \ref{Algo:EpsApprox} on input $\samp$, $\ch$, $\lo = \hi/2^i$, $\hi$, $\delta/2^i$.
   \begin{enumerate}
      \item If the result is non-zero, run Algorithm \ref{Algo:EpsApprox} on input $\samp$, $\ch$, $\lo/1250$, $\hi$, $\eps$, $\delta/2^{i+1}$ and \underline{output} its result as $\muts$.
      \item Else, set $i = i + 1$ and go to Step 2.
   \end{enumerate}
\end{enumerate}
}

\begin{proof}[Proof of Theorem \ref{Thm:ExpectedEpsApprox}]
  We show that Algorithm \ref{Algo:EpsApproxMu} satisfies the properties specified in Theorem \ref{Thm:ExpectedEpsApprox}.

  Suppose that $\ch \geq \phis / \mus$ and $\hi > \mus$. Since Remark \ref{Rem:RemoveL} also applies to Algorithm \ref{Algo:EpsApprox}, the probability that Algorithm \ref{Algo:EpsApproxMu} stops with $\lo \geq 1250 \mus$ is at most $\sum_{i=1}^{\lfloor\log(\hi/(1250\mus))\rfloor} \delta/2^i$. On the other hand, if $\lo < 1250 \mus$ at Step 2.(a) then, according to Theorem \ref{Thm:EpsApprox}, the output $\muts$ satisfies $|\muts - \mus| \leq \eps \mus$ with probability $1-\delta/2^{i+1}$. Consequently, the output is correct with probability at least $1 - \sum_{i=1}^{\infty} \delta/2^i \geq \delta$.

  According to Remark \ref{Rem:RemoveL}, when $\lo < \mus$ the probability that Step 2 computes a non-zero value is at least $1 - \delta/2^i$. Thus, Algorithm \ref{Algo:EpsApproxMu} uses
    \[\so{\ch \cdot \eps^{-1}\log(1/\delta) \left(\sum\nolimits_{i=1}^{\lfloor\log(\hi/\mus)\rfloor} i^2 + \sum\nolimits_{i = \lfloor\log(\hi/\mus)\rfloor}^{\infty} i^2 \cdot \delta/2^i \right)} = \so{\ch \cdot \eps^{-1} \log^3(\hi/\mus)\log(1/\delta) }\]
  \qss in expectation.
\end{proof}

  \section{\texorpdfstring{Approximating the mean when $\ch$ is implicit}{Approximating the mean when DeltaS is implicit}}
  \label{App:appDecreasingFunction}

We show how to approximate the mean $\mus$ of a quantum sampler $\samp$ given a non-increasing function $f$ such that $f(\mus) \geq \phis/\mus$. Our result combines Algorithm \ref{Algo:EpsApprox} (or Algorithm \ref{Algo:basicEpsApprox}) with a new stopping rule that is based on the Markov-like inequality ``$\muts \leq (1+2\pi)^2 \mus$'' of Theorem \ref{Thm:EpsApprox}.

\boxalgo{Algo:EpsApproxDec}{$\eps-$approximation of the mean of a quantum sampler $\samp$ for implicit $\ch$.}{
{\bf Input:} a sampler $\samp$, a non-increasing function $f$ such that $f(\mus) \geq \phis / \mus$, two values $0 < \lo < \hi$, two reals $0 < \eps, \delta < 1/2$. \\
{\bf Output:} an estimate $\muts$ of $\mus$.

\begin{enumerate}
  \item Set $\md = 2\hi$, $\ch = f(\md)$ and $\mut = 0$.
  \item While $\mut < \md/6$ and $\md \geq \lo/2$:
  \begin{enumerate}
    \item Set $\md = \md/2$ and $\ch = f(\md)$.
    \item Run Algorithm \ref{Algo:EpsApprox} on input $\samp$, $\ch$, $\lo$, $\hi$, $\eps' = 5/6$, $\delta' = \frac{\delta}{2 \left(2+\log\left(\frac{\hi}{\lo}\right) \right)}$. Denote the result by $\mut$.
  \end{enumerate}
  \item If $M < L/2$ then \underline{output} $\muts = 0$.
  \item Else, run Algorithm \ref{Algo:EpsApprox} on input $\samp$, $\ch = f\left(\md/(6(1+2\pi)^2)\right)$, $\lo$, $\hi$, $\eps$, $\delta/2$ and \underline{output} its result as $\muts$.
\end{enumerate}
}

\begin{thm}
\label{Thm:EpsApproxDec}
  If $\lo \leq \mus < \hi$ then the output $\muts$ of Algorithm \ref{Algo:basicEpsApprox} satisfies $|\muts - \mus| \leq \eps \mus$ with probability $1-\delta$.  Moreover, for any $\lo$ it satisfies $\muts \leq (1+2\pi)^2 \mus$ with probability $1-\delta$. The number of \qss used by the algorithm is
    \[\so{f\left(\frac{\max(L/4,2^{-T}\mus)}{6(1+2\pi)^2}\right) \cdot \eps^{-1} \log\left(\frac{\hi}{\lo}\right) \log\left(\frac{1}{\delta}\right)}\]
  for some integer random variable $T$ such that $\Pb(T = 1) \geq 1-\delta$ and $\Pb(T = \ell) \leq \delta^{\ell}$ for all $\ell > 1$.
\end{thm}

\begin{proof}
  Assume first that $L \leq \mus$. According to Theorem \ref{Thm:EpsApprox}, the estimate $\mut$ computed at Step 2.(b) of Algorithm \ref{Algo:EpsApproxDec} satisfies $\mut \leq (1+2\pi)^2 \mus$ with probability $1-\delta'$. Consequently, when $\md > 6(1+2\pi)^2 \mus$, we have $\mut < \md/6$ with probability $1-\delta'$. On the other hand, when $\md \leq \mus$, since $\ch = f(\md) \geq \phis/\mus$ the value $\mut$ satisfies $|\mut - \mus| \leq (5/6) \cdot \mus$ with probability $1-\delta'$ (by Theorem \ref{Thm:EpsApprox}). In particular, it implies $\mut \geq \mus/6 \geq \md/6$ with probability $1-\delta'$. Using these two points, we conclude that the first time Step 2.(b) of Algorithm \ref{Algo:EpsApproxDec} obtains $\mut \geq \md/6$ happens for $\md \in [\mus/2, 6 (1+2\pi)^2 \mus]$, with probability at least $(1-\delta')^{1+\log(\hi/(\mus/2))} > 1-\delta/2$. In this case, $\ch \geq \phis/\mus$ at Step 4 of the algorithm, and the output $\muts$ satisfies $|\mut - \mus| \leq \eps \mus$ with probability $1-\delta/2$ (by Theorem \ref{Thm:EpsApprox}). The total success probability is $(1-\delta/2)^2 \geq 1-\delta$.

  If $\lo > \mus$, this may only increase the probability to stop at Step 3 and output $\muts = 0$. If Step 4 is executed, we still have $\muts \leq (1+2\pi)^2 \mus$ with probability $1-\delta$, as a consequence of Theorem \ref{Thm:EpsApprox}.

  We analyse the number of \qss used in the algorithm. The value taken by $M$ at Step 4 satisfies $\md \geq \mus/2$ with probability at least $1-\delta$, and $2^{-\ell}\mus > \md \geq 2^{-(\ell+1)} \mus$ with probability at most $\delta^{\ell}$ (for any $\ell \geq 1$). Moreover, the total number of \qss used in Algorithm \ref{Algo:EpsApproxDec} is dominated (up to a polylogarithmic factor in $\hi/\lo$) by the number of \qss used at Step 4, that is $\so{f\left(\frac{2^{-\ell}\mus}{12(1+2\pi)^2}\right) \cdot \eps^{-1} \log\left(\frac{\hi}{\lo}\right) \log\left(\frac{1}{\delta}\right)}$ when $M \geq 2^{-\ell}\mus$. The smallest possible value for $M$ at Step 4 is $L/4$. Thus, the total number of \qss is $\so{f\left(\frac{\max(L/4,2^{-T} \mus)}{12(1+2\pi)^2}\right) \cdot \eps^{-1} \log\left(\frac{\hi}{\lo}\right) \log\left(\frac{1}{\delta}\right)}$, where $T = 1$ with probability at least $1-\delta$ and $T = \ell$ with probability at most $\delta^{\ell}$, for all $\ell \geq 2$.
\end{proof}

We simplify the above statement when the function $f$ is of the form $f : x \mapsto A/x^{\alpha}$ for some $A, \alpha > 0$ (this result is sufficient for our applications in Section \ref{Sec:graphParam}).

\begin{cor}
\label{Cor:EpsApproxDecSimpl}
  If $\lo \leq \mus < \hi$ and $f : x \mapsto A/x^{\alpha}$ for some reals $A, \alpha > 0$ with $\delta < 2^{-2\alpha}$, then the output $\muts$ of Algorithm \ref{Algo:basicEpsApprox} satisfies $|\muts - \mus| \leq \eps \mus$ with probability $1-\delta$.  Moreover, for any $\lo$ it satisfies $\muts \leq (1+2\pi)^2 \mus$ with probability $1-\delta$. The algorithm uses
    \[\so{f(\max(L,\mus)) \cdot \eps^{-1} \log\left(\frac{\hi}{\lo}\right) \log\left(\frac{1}{\delta}\right)}\]
  \qss in expectation (both for the $\ell_1$ and $\ell_2$ average).
\end{cor}

\begin{proof}
  The average (for the $\ell_1$ norm) number of \qss used in Algorithm \ref{Algo:EpsApproxDec} is
    \[\so{\sum_{\ell = 0}^{\infty} \delta^{\ell} \cdot f\left(\frac{\max(L/4,2^{-(\ell+1)}\mus)}{6(1+2\pi)^2}\right) \cdot \eps^{-1}  \log\left(\frac{\hi}{\lo}\right) \log\left(\frac{1}{\delta}\right)}\]
  Since $f : x \mapsto A/x^{\alpha}$ for some $A, \alpha > 0$, it becomes $\so{\frac{A}{\max(L,\mus)^{\alpha}} \cdot \eps^{-1}  \log\left(\frac{\hi}{\lo}\right) \log\left(\frac{1}{\delta}\right)\bigr)}$ when $\delta < 2^{-\alpha}$. Similarly, for the $\ell_2$ norm, the average number of \qss used in Algorithm \ref{Algo:EpsApproxDec} is
    \[\so{\left(\sum_{\ell = 0}^{\infty} \delta^{\ell} \cdot \left(f\left(\frac{\max(L/4,2^{-(\ell+1)}\mus)}{6(1+2\pi)^2}\right) \cdot \eps^{-1}  \log\left(\frac{\hi}{\lo}\right) \log\left(\frac{1}{\delta}\right)\right)^2\right)^{1/2}}\]
  which becomes $\so{\frac{A}{\max(L,\mus)^{\alpha}} \cdot \eps^{-1} \log\left(\frac{\hi}{\lo}\right) \log\left(\frac{1}{\delta}\right)\bigr)}$ when $\delta < 2^{-2\alpha}$.
\end{proof}

  \section{Approximating the mean of variable-time samplers}
  \label{App:appVariableTime}

\begin{defn}[Variable-time algorithm \cite{Amb10b,CGJ18}]
  \label{Def:varTime}
  Consider two Hilbert spaces $\Hil_F = \otimes_{i=1}^m \Hil_{F_i}$ (for some integer $m$) and $\Hil_C$, where each $\Hil_{F_i}$ is equipped with a standard basis $\{\ket{\sto},\ket{\cont}\}$. We say that a unitary $U$ acting on $\Hil_F \otimes \Hil_C$ is a \emph{variable-time algorithm with stopping times $t_1 < \dots < t_m$} if it can be decomposed as a product of unitary operators $U = U_m \cdots U_1$, such that each $U_i$ has time complexity $\tmax(U_i) = t_i - t_{i-1}$ (where $t_0 = 0$) and acts on $\Hil_{F_i} \otimes \Hil_C$ controlled on the first $(i-1)^{th}$ registers being $\ket{\cont}^{\otimes i-1} \in \otimes_{j=1}^{i-1} \Hil_{F_j}$. The \emph{probability to stop at step $i$} is defined as
    \[p_{\sto, i} = \| \Pi_{\sto} (U_i \cdots U_1 \ket{\mathit{init}}) \|^2 - \| \Pi_{\sto} (U_{i-1} \cdots U_1 \ket{\mathit{init}}) \|^2\]
  where $\ket{\mathit{init}} = \ket{\cont}^{\otimes m} \ket{0} \in \Hil_F \otimes \Hil_C$ and $\Pi_{\sto}$ is the projector on $\spa{\ket{\cont}^{\otimes m}}^{\perp} \otimes \Hil_C$ (i.e. on the states containing $\ket{\sto}$). The \emph{$\ell_2$-average running time} of $U$ is defined as $\talgs(U) = (\sum_{i=1}^m p_{\sto,i} \cdot t_i^2)^{1/2}$.
\end{defn}

The previous definition expresses the fact that some branches of computation may stop earlier than the others. When a branch is completed at time $t_i$, the corresponding register in $\Hil_{F_i}$ is set to $\ket{\sto}$, and this part of the state cannot be changed afterward. Ambainis \cite{Amb10b} studied the question of quantum search and amplitude amplification for variable-time unitaries $U = U_m \cdots U_1$. We extend this work by developing the following \emph{variable-time amplitude estimation} algorithm.\footnote{We have been aware, during the redaction of this paper, of a similar result recently obtained in \cite{CGJ18} with time complexity
  $\bo{(\tmax(U) + t \cdot \talgs) \eps^{-1} \cdot  \log^3(\tmax(U))\log\left(\log(\tmax(U))/\delta\right)}$ that is too large for our applications.}

\begin{thm}
  \label{Thm:VarTime}
  Consider two Hilbert spaces $\Hil_F = \otimes_{i=1}^m \Hil_{F_i}$ (for some integer $m$) and $\Hil_C$, where each $\Hil_{F_i}$ is equipped with a standard basis $\{\ket{\sto},\ket{\cont}\}$. There is a quantum algorithm that takes as input a variable-time algorithm $U = U_m \cdots U_1$ on $\Hil_F \otimes \Hil_C$, an orthogonal projector $\Pi_C$ on $\Hil_C$, two reals $t, \talgs > 1$, and two reals $0 < \eps, \delta < 1$. If $\talgs \geq \talgs(U)$, then the algorithm outputs an estimate $\pt$ of $p = \bracket{\psi}{\Pi}{\psi}$, where $\Pi = (I_{\Hil_F} - (\proj{\cont})^{\otimes m}) \otimes \Pi_C$ and $\ket{\psi} = U (\ket{\cont}^{\otimes m} \ket{0})$, such that
  \medskip
  \begin{center}
    \emph{(1)} \ $\pt \leq 2 \cdot p$, \ for any $t$; \qquad \emph{(2)} \ $|\pt-p| \leq \eps \cdot p$, \ when $t \geq \frac{2}{\sqrt{p}}$; \qquad \emph{(3)} \ $\pt = 0$, \ when $t < \frac{1}{\sqrt{2p}}$.
  \end{center}
  \medskip
  with probability $1-\delta$. The time complexity of this algorithm is
  $$
    \bo{\left(\min\left(\tmax(U), t \cdot \talgs \eps^{-1/2}\right) + t \cdot\talgs\right)\eps^{-1} \cdot \log^4(\tmax(U))\log\left(\frac{\log(\tmax(U))}{\delta}\right)}.
  $$
\end{thm}

Using this new result in place of the standard amplitude estimation in Algorithm \ref{Algo:EpsApprox}, we obtain directly the following result.

\begin{thm}
\label{Thm:VarEpsApprox}
  There is an algorithm that, given a variable-time sampler $\samp$, an integer $\ch$, two values $0 < \lo < \hi$, a real $\talgs \geq 1$, and two reals $0 < \eps, \delta < 1$, outputs an estimate $\muts$ of $\mus$. If $\ch \geq \phis / \mus$, $\talgs \geq \talgs(\samp)$ and $\lo < \mus < \hi$, then it satisfies $|\muts - \mus| \leq \eps \mus$ with probability $1-\delta$. Moreover, for any $\ch, \lo, \hi, \talgs$ it satisfies $\muts \leq 2 \cdot \mus$ with probability $1-\delta$. The \emph{time complexity} of this algorithm is
    \[\so{\ch \left(\eps^{-2} + \log\left(\frac{\hi}{\lo}\right)\right) \cdot \talgs \cdot \log^4(\tmax(\samp)) \log\left(\frac{ 1}{\delta}\right)}.\]
\end{thm}

The rest of this section is dedicated to the proof of Theorem \ref{Thm:VarTime}. Our approach (Algorithms \ref{Algo:ApproxPacc} and \ref{Algo:VarTime}) consists in estimating at each intermediate time step $t_i$ of $U$ a multiplicative portion $\pt_i$ of $p$ (the final estimate $\pt$ of $p$ being the product of the $\pt_i$'s). To this end, we apply the amplitude estimation algorithm on two particular state generation algorithms $(\balg_i)_i$ and $(\alg_i)_i$ (Algorithms \ref{Algo:StateGenB} and \ref{Algo:StateGenA}) originating from the work of Ambainis \cite{Amb10b}.


\subsection{Preliminaries}

We need a modified version of the amplitude estimation algorithm that does not need input time parameter.

\begin{prop}[{\cite[Theorem 15]{BHMT02}}]
  \label{Prop:amplSafe}
  There is a quantum algorithm, denoted $\Xamps$, that takes as input a unitary operator $U$, an orthogonal projector $\Pi$, and two reals $0 < \eps, \delta < 1$. With probability $1-\delta$, this algorithm outputs an estimate $\pt = \amps{U,\Pi,\eps,\delta}$ satisfying $|\pt-p| \leq \eps p$ and runs in time
    \[\bo{\frac{\tmax(U)}{\eps \sqrt{p}} \cdot \log\left(\frac{1}{\delta}\right)}\]
  where $p = \bra{\psi} \Pi \ket{\psi}$ and $\ket{\psi} = U \ket{0}$.
\end{prop}

We also use the following careful analysis of the amplitude amplification algorithm.

\begin{prop}[{\cite[Lemma 5.2]{AA05}}]
  \label{Prop:Amplify}
  Let $\Hil$ be some Hilbert space.
  Let $U$ be a unitary operator and $\Pi$ an orthogonal projector on $\Hil$. Denote $p = \bra{\psi} \Pi \ket{\psi}$ where $\ket{\psi} = \sqrt{p} \ket{\psi_{\Pi}} + \sqrt{1-p} \ket{\psi_{\Pi^{\perp}}} = U \ket{0}$ and $\ket{\psi_{\Pi}}$, $\ket{\psi_{\Pi^{\perp}}}$ are two unit vectors invariant by $\Pi$ and $\Pi^{\perp}$ respectively. Given an integer $t$ such that
    \[t \leq \frac{\pi}{4 \arcsin \sqrt{p}} - \frac{1}{2}\]
  the Amplitude Amplification algorithm \cite[Section 2]{BHMT02} on input $(U, \Pi, t)$ outputs in time $\bo{t}$ the description of a quantum circuit $\amplify{U,\Pi,t}$ acting on $\Hil$ such that
    \[\amplify{U,\Pi,t} \ket{0} = \sqrt{p'} \ket{\psi_{\Pi}} + \sqrt{1-p'} \ket{\psi_{\Pi^{\perp}}}\]
  where
    \[p' \geq \left(1 - \frac{(2t+1)^2}{3}p\right)(2t+1)^2 p\]
  Moreover, $\amplify{U,\Pi,t}$ runs in time $\bo{t \cdot \tmax(U)}$.
\end{prop}


\subsection{Notations}
For clarity, and without loss of generality, we assume that each intermediate state $\ket{\psi_i} = U_i \cdots U_1 \ket{\mathit{init}}$ of the variable-time algorithm $U = U_m \cdots U_1$ can be written as
  $$
   \begin{cases}
    \ket{\psi_i} = \sqrt{p_{\rej, \leq i}} \ket{\sto^0_i} \ket{\psi^0_i}\ket{0} + \sqrt{p_{\acc, \leq i}} \ket{\sto^1_i}\ket{\psi^1_i}\ket{1} + \sqrt{p_{\sto, > i}} \ket{cont_i}\ket{\psi^2_i}\ket{2} , & \text{for $i < m$;} \\
    \ket{\psi_m} = \ket{\psi} = \sqrt{p_{\rej}} \ket{\sto^0_m}\ket{\psi^0}\ket{0} + \sqrt{p_{\acc}} \ket{\sto^1_m}\ket{\psi^1}\ket{1} , & \text{where $p = p_{\acc}$.}
   \end{cases}
  $$
for some unit vectors $\ket{\sto^1_i}, \ket{\sto^0_i} \in \spa{\ket{\cont}^{\otimes i}}^{\perp} \otimes_{j=i+1}^m \Hil_{F_j} \otimes \Hil_C$, $\ket{cont_i} \in \spa{\ket{\cont}^{\otimes i}} \otimes_{j=i+1}^m \Hil_{F_j} \otimes \Hil_C$, $\ket{\psi^0_i}, \ket{\psi^1_i}, \ket{\psi^2_i}, \ket{\psi^0}, \ket{\psi^1} \in \Hil_{C'}$ where $\Hil_C = \Hil_{C'} \otimes \C^3$, and some probabilities $p_{\acc, \leq i}$, $p_{\rej, \leq i}$, $p_{\sto, > i}$, $p_{\acc}$, $p_{\rej}$. The last register indicates if the computation is not finished (value $2$), if it is finished and corresponds to the \emph{accepting part} whose amplitude has to be estimated (value $1$), or if it is finished and corresponds to the \emph{rejecting part} (value $0$). The proportion $1 - p_{\sto, > i}$ of computation that is finished at step $i$ is decomposed as $p_{\acc, \leq i}$ for the accepting part and $p_{\rej, \leq i}$ for the rejecting part. We assume that all the computations are finished at step $m$ (i.e. $p_{\sto, >m} = 0$, $p_{\acc, \leq m} = p_{\acc} = p$ and $p_{\rej, >m} = p_{\rej}$). We also denote $p_{\rej, \leq 0} = p_{\acc, \leq 0} = 0$, $p_{\sto, >0} = 1$. Finally, we define the following two projectors on $\Hil_F \otimes \Hil_C$:
  $$
   \begin{cases}
     \Pi_1 = I_{\Hil_F \otimes \Hil_{C'}} \otimes \proj{1} ; \\
     \Pi_{1,2} = I_{\Hil_F \otimes \Hil_{C'}} \otimes (\proj{1}+\proj{2}) . \\
   \end{cases}
  $$


\subsection{State generation algorithms}

We recall the definition of the state generation algorithms $(\balg_i)_i$ and $(\alg_i)_i$ from \cite{Amb10b}.

\boxalgo{Algo:StateGenB}{state generation algorithm $\textup{\textsf{Gen}}_{\balg}$.}{
{\bf Input:} a variable-time algorithm $U = U_m \cdots U_1$ with stopping times $t_1 < \dots < t_m$, a step $i \in \{1,\dots,m\}$, a sequence of estimates $(\bt_k)_{1 \leq k \leq i-1}$. \\
{\bf Output:} a state generation algorithm $\balg_i = \genb{U,i,(\bt_k)_{1 \leq k \leq i-1}}$.

  \begin{enumerate}
    \item If $i = 1$, \underline{output} $\balg_1 = U_1$.
    \item If $i > 1$, \underline{output} $\balg_i = U_i \alg_{i-1}$ where $\alg_{i-1} = \gena{U,i-1,(\bt_k)_{1 \leq k \leq i-1}}$.
  \end{enumerate}
}

\boxalgo{Algo:StateGenA}{state generation algorithm $\textup{\textsf{Gen}}_{\alg}$.}{
{\bf Input:} a variable-time algorithm $U = U_m \cdots U_1$ with stopping times $t_1 < \dots < t_m$, a step $i \in \{1,\dots,m\}$, a sequence of estimates $(\bt_k)_{1 \leq k \leq i}$. \\
{\bf Output:} a state generation algorithm $\alg_i = \gena{U,i,(\bt_k)_{1 \leq k \leq i}}$.

\begin{enumerate}
  \item Set $\balg_i = \genb{U,i,(\bt_k)_{1 \leq k \leq i-1}}$. 
  \item If $\bt_i > \frac{1}{9m}$, \underline{output} $\alg_i = \balg_i$.
  \item If $\bt_i \leq \frac{1}{9m}$, \underline{output} $\alg_i = \amplify{\balg_i,\Pi_{1,2},k}$ for the smallest $k$ satisfying $1/(9m) \leq (2k+1)^2 \bt_i \leq 1/m$.
\end{enumerate}
}

We let $\ket{\psi_{\balg_i}} = \balg_i \ket{\mathit{init}}$ and $\ket{\psi_{\alg_i}} = \alg_i \ket{\mathit{init}}$ denote the states generated by the $(\balg_i)_i$ and $(\alg_i)_i$ algorithms respectively. The goal of the $(\alg_i)_i$ algorithms is to amplify at each intermediate step $i$ the amplitude of the potentially accepting part $\sqrt{p_{\acc, \leq i}} \ket{\sto^1_i}\ket{\psi^1_{\leq i}} \ket{1} + \sqrt{p_{\sto, > i}} \ket{cont_i}\ket{\psi^2_{> i}} \ket{2}$ into $\ket{\psi_{\balg_i}}$ from $b_i = \| \Pi_{1,2} \ket{\psi_{\balg_i}} \|^2$ to $a_i = \| \Pi_{1,2} \ket{\psi_{\alg_i}} \|^2 \geq \max(b_i,\om{1/m})$. The goal of the $(\balg_i)_i$ algorithms is to continue the execution of $U$: $\ket{\psi_{\balg_{i+1}}} = U_{i+1} \ket{\psi_{\alg_i}}$. Below we summarize the main results from \cite{Amb10b} we need about these algorithms.

\begin{prop}[\cite{Amb10b}]
  \label{Prop:timeAi}
  Consider a variable-time algorithm  $U = U_m \cdots U_1$ with stopping times $t_1 < \dots < t_m$, a step $i \in \{1,\dots,m\}$ and a sequence of estimates $(\bt_k)_{1 \leq k \leq i}$. For each $1 \leq j \leq i$, denote $\balg_{j} = \genb{U,j,(\bt_k)_{1 \leq k \leq j-1}}$, $\alg_{j} = \gena{U,j,(\bt_k)_{1 \leq k \leq j}}$, and let $b_j = \| \Pi_{1,2} (\balg_j \ket{\mathit{init}})\|^2$, $a_j = \| \Pi_{1,2} (\alg_j \ket{\mathit{init}})\|^2$. We have that
    \begin{equation}
      \label{Equ:ampl}
      b_i = a_{i-1} \frac{1 - p_{\rej, \leq i}}{1 - p_{\rej, \leq i-1}}
    \end{equation}
  where $a_0 = 0$. Moreover, if $|\bt_j - b_j| \leq b_j/(3m)$ for all $1 \leq j \leq i$, then the running time $\tmax(\alg_i)$ of $\alg_{i}$ is
   \[\tmax(\alg_i) \leq C \sqrt{m} \left(t_i + {i} \frac{\talgs(U)}{\sqrt{1-p_{\rej, \leq i}}}\right)\]
  for some constant $C$, and
   \[a_i \geq \left(1-\frac{1}{3m}\right) \frac{1}{9m}.\]
\end{prop}

\subsection{Variable-time amplitude estimation algorithm}

We describe the two algorithms that constitute our variable-time amplitude estimation algorithm. First, we show how to approximate $p_{\acc,\leq i}$ for any step $i$ (Algorithm \ref{Algo:ApproxPacc}). Then, we describe the algorithm proving Theorem \ref{Thm:VarTime} (Algorithm \ref{Algo:VarTime}). Our results rely on the following consequence of Equation \ref{Equ:ampl}.

\begin{lem}
  \label{Lem:collision}
Using the notations of Proposition \ref{Prop:timeAi}, we have that
  \[p_{\acc, \leq i} = b_1 \cdot \prod_{j=2}^{i-1} \frac{b_j}{a_{j-1}} \cdot \frac{b_{i,1}}{a_{i-1}}\]
where $b_{i,1} = \|\Pi_1 (\balg_i \ket{\mathit{init}})\|^2 = a_{i-1} \frac{p_{\acc, \leq i}}{1 - p_{\rej, \leq {i-1}}}$.
\end{lem}

\boxalgo{Algo:ApproxPacc}{estimation of $p_{\acc,\leq i}$.}{
{\bf Input:} a variable-time algorithm $U = U_m \cdots U_1$, a step $i \in \{1,\dots,m\}$, two reals $0 < \eps, \delta < 1$. \\
{\bf Output:} an estimate $\pt_{\acc,\leq i}$ of $p_{\acc,\leq i}$.

\begin{enumerate}
  \item For $j = 1,\dots,i-1$:
    \begin{enumerate}
      \item Set $\balg_j = \genb{U,j,(\bt_k)_{1 \leq k \leq j-1}}$ and compute $\bt_j = \amps{\balg_j,\Pi_{1,2},\frac{\eps}{4m},\frac{\delta}{2m}}$.
      \item Set $\alg_j = \gena{U,j,(\bt_k)_{1 \leq k \leq j}}$ and compute $\at_j = \amps{\alg_j,\Pi_{1,2},\frac{\eps}{8m},\frac{\delta}{2m}}$.
    \end{enumerate}
  \item Set $\balg_i = \genb{U,i,(\bt_k)_{1 \leq k \leq i-1}}$ and compute $\bt_{i,1} = \amps{\balg_i,\Pi_1,\frac{\eps}{4m},\frac{\delta}{2m}}$.
  \item \underline{Output} $\pt_{\acc,\leq i} = \bt_1 \cdot \prod_{j=2}^{i-1} \frac{\bt_j}{\at_{j-1}} \cdot \frac{\bt_{i,1}}{\at_{i-1}}$.
\end{enumerate}
}

\begin{prop}
  \label{Prop:ApproxPacc}
  With probability $1-\delta$, Algorithm \ref{Algo:ApproxPacc} outputs an estimate $\pt_{\acc,\leq i}$ satisfying $|\pt_{\acc,\leq i} - p_{\acc,\leq i}| \leq \eps p_{\acc,\leq i}$ and runs in time $\bo{\frac{m^3}{\eps} \sqrt{\frac{1 - p_{\rej,\leq i}}{p_{\acc, \leq i}}} \left(t_i + i \frac{\talgs(U)}{\sqrt{1-p_{\rej, \leq i}}}\right) \log\left(\frac{m}{\delta}\right)}$.
\end{prop}

\begin{proof}
  Using Proposition \ref{Prop:amplSafe}, together with a union bound over all the calls to $\Xamps$ in Algorithm \ref{Algo:ApproxPacc}, we can assume with probability $1-\delta$ that (for all $j$) $\bt_j$ and $\bt_{i,1}$ are $\frac{\eps}{4m}$-approximations of $b_j$ and $b_{i,1}$ respectively, and $\at_j$ is an  $\frac{\eps}{8m}$-approximation of $a_j$ (which implies $\left|\frac{1}{\at_j} - \frac{1}{a_j}\right| \leq \frac{\eps}{4m} \cdot \frac{1}{a_j}$). Consequently,
    \begin{equation*}
      \bt_1 \cdot \prod_{j=2}^{i-1} \frac{\bt_j}{\at_{j-1}} \cdot \frac{\bt_{i,1}}{\at_{i-1}}
           \leq \left(1 + \frac{\eps}{4m}\right)^{2i} b_1 \cdot \prod_{j=2}^{i-1} \frac{b_j}{a_{j-1}} \cdot \frac{b_{i,1}}{a_{i-1}}
           \leq \left(1+\frac{4i}{4m}\eps\right) \cdot p_{\acc, \leq i}
           \leq (1+\eps) p_{\acc, \leq i}
    \end{equation*}
 where we used Lemma \ref{Lem:collision} and the inequalities $1+x \leq e^x$ and $e^y - 1 \leq 2y$ (for $y \in [0,1$]). On the other hand,
    \begin{equation*}
      \bt_1 \cdot \prod_{j=2}^{i-1} \frac{\bt_j}{\at_{j-1}} \cdot \frac{\bt_{i,1}}{\at_{i-1}}
           \geq \left(1 - \frac{\eps}{4m}\right)^{2i} b_1 \cdot \prod_{j=2}^{i-1} \frac{b_j}{a_{j-1}} \cdot \frac{b_{i,1}}{a_{i-1}}
           \geq \left(1-\frac{2i}{4m} \eps\right) \cdot p_{\acc, \leq i}
           \geq (1-\eps) p_{\acc, \leq i}
    \end{equation*}
    where we used Lemma \ref{Lem:collision} and Bernoulli's inequality. Thus, $|\pt_{\acc,\leq i} - p_{\acc,\leq i}| \leq \eps p_{\acc,\leq i}$.

    We analyse the time complexity of the algorithm. Using the same union bound as above we can assume with probability $1-\delta$ that (for all $j$) Step 1.(a) runs in time $\bo{\frac{m}{\eps \sqrt{b_j}} \tmax(\balg_j) \log(m/\delta)}$, Step 1.(b) runs in time $\bo{\frac{m}{\eps \sqrt{a_j}} \tmax(\alg_j) \log(m/\delta)}$ and Step 2 runs in time $\bo{\frac{m}{\eps \sqrt{b_{i,1}}} \tmax(\balg_i) \log(m/\delta)}$. Moreover, observe that if $\bt_j > \frac{1}{9m}$ then $a_j = b_j$ and $\tmax(\alg_j) \geq \tmax(\balg_j)$, and if $\bt_j \leq \frac{1}{9m}$ then $\tmax(\alg_j) = \om{\sqrt{\frac{a_j}{b_j}} \tmax(\balg_j)}$, by definitions of $(\balg_j)_j$ and $(\alg_j)_j$. In both cases we obtain $\frac{\tmax(\balg_j)}{\sqrt{b_j}} = \bo{\frac{\tmax(\alg_j)}{\sqrt{a_j}}}$. Similarly, $\frac{\tmax(\balg_i)}{\sqrt{b_{i,1}}} = \bo{\sqrt{\frac{b_i}{a_i b_{i,1}}} \tmax(\alg_i)}$. Consequently, using Proposition \ref{Prop:timeAi}, the total time complexity is
      $
       \bo{\left(\sum_{j=1}^{i-1} \frac{m}{\eps \sqrt{a_j}} \tmax(\alg_j) + \sqrt{\frac{b_i}{a_i b_{i,1}}} \tmax(\alg_i)\right)
      \log\left(\frac{m}{\delta}\right)}
     = \bo{\frac{m^3}{\eps} \sqrt{\frac{1 - p_{\rej,\leq i}}{p_{\acc, \leq i}}} \left(t_i + i \frac{\talgs(U)}{\sqrt{1-p_{\rej, \leq i}}}\right) \log\left(\frac{m}{\delta}\right)}
      $.
\end{proof}

In the following, we make the basic assumption (also used in \cite{Amb10b,CGJ18}) that $U = U_m \cdots U_1$ has stopping times $t_j = 2^j$, for $j = 1, \dots, m$ and $m = \log(\tmax(U))$.

\boxalgo{Algo:VarTime}{estimation of $p_{\acc}$.}{
{\bf Input:} a variable-time algorithm $U = U_m \cdots U_1$ with stopping times $t_j = 2^j$ ($1 \leq j \leq m$), an integer $t$, a value $\talgs \geq \talgs(U)$, two reals $0 < \eps, \delta < 1$. \\
{\bf Output:} an estimate $\pt_{\acc}$ of $p_{\acc}$.

\begin{enumerate}
  \item Set $i = \min\left(m,\lceil \log(t \eps^{-1/2} \cdot \talgs)\rceil\right)$ and $t' = 2D \frac{m^3}{\eps} (t_i + i \cdot t \cdot \talgs) \log\left(\frac{m}{\delta}\right)$, where $D$ is the constant hidden in the $\bo{.}$ notation of Proposition \ref{Prop:ApproxPacc}.
  \item Run Algorithm \ref{Algo:ApproxPacc} with input $U$, $i$, $\eps/2$, $\delta$ for at most $t'$ computation steps.
  \begin{enumerate}
    \item If the computation has not ended after $t'$ steps, stop it and \underline{output} $\pt_{\acc} = 0$.
    \item Else, let $\pt_{\acc,\leq i}$ denote the result of Algorithm \ref{Algo:ApproxPacc}. If $\pt_{\acc,\leq i} = 0$ or $t < 1/\sqrt{\pt_{\acc,\leq i}}$ then \underline{output} $\pt_{\acc} = 0$, else \underline{output} $\pt_{\acc} = \pt_{\acc,\leq i}$.
  \end{enumerate}
\end{enumerate}
}

\begin{proof}[Proof of Theorem \ref{Thm:VarTime}]
  We show that Algorithm \ref{Algo:VarTime} satisfies the statements of Theorem \ref{Thm:VarTime}.

  Assume first that $t \geq \frac{2}{\sqrt{p_{\acc}}}$. Since $\talgs \geq \talgs(U) \geq \sqrt{p_{\sto, > i} \cdot t_i^2} = \sqrt{p_{\sto, > i} \cdot 2^{2i}}$ for all $i$, by choosing $i = \min\left(m,\lceil \log(t \eps^{-1/2} \cdot \talgs)\rceil\right)$ we obtain $p_{\sto, > i} \leq \talgs^2/t_i \leq (\eps/4) \cdot p_{\acc}$. Thus $p_{\acc, \leq i}$ satisfies
    \[p_{\acc} \geq p_{\acc, \leq i} \geq p_{\acc} - p_{\sto, > i} \geq (1-\eps/4) \cdot p_{\acc}\]
  and $1 - p_{\rej,\leq i} \leq p_{\acc} + p_{\sto, > i} \leq 2 p_{\acc}$.
  It implies
    $D \frac{m^3}{\eps} \sqrt{\frac{1 - p_{\rej,\leq i}}{p_{\acc, \leq i}}} \left(t_i + i \frac{\talgs(U)}{\sqrt{1-p_{\rej, \leq i}}}\right) \log\left(\frac{m}{\delta}\right)
    < t'$.
  Consequently, according to Proposition \ref{Prop:ApproxPacc}, with probability $1-\delta$ the computation does not stop at Step 2 and $\pt_{\acc,\leq i}$ satisfies $|\pt_{\acc,\leq i} - p_{\acc,\leq i}| \leq (\eps/2) \cdot p_{\acc,\leq i}$. In this case, using the triangle inequality, we have $|\pt_{\acc,\leq i} - p_{\acc}| \leq \eps \cdot p_{\acc}$ and $1/\sqrt{\pt_{\acc,\leq i}} \leq \sqrt{2/p_{\acc}} \leq t$.

  Assume now that $t < \frac{2}{\sqrt{p_{\acc}}}$. According to Proposition \ref{Prop:ApproxPacc}, the output $\pt_{\acc,\leq i}$ of Algorithm \ref{Algo:ApproxPacc} satisfies $\pt_{\acc,\leq i} \leq (1+\eps/2) p_{\acc,\leq i} \leq 2 p_{\acc}$ with probability $1-\delta$. Since the output $\pt_{\acc}$ of Algorithm \ref{Algo:VarTime} is either $0$ or $\pt_{\acc,\leq i}$, it also satisfies $\pt_{\acc} \leq 2 p_{\acc}$ with probability $1-\delta$. Finally, if $t < \frac{1}{\sqrt{2p_{\acc}}}$ and $0 \neq \pt_{\acc,\leq i} \leq 2 p_{\acc}$ then $t < \frac{1}{\sqrt{\pt_{\acc,\leq i}}}$ and $\pt_{\acc} = 0$.
\end{proof}

  \section{Making streaming algorithms reversible}
  \label{App:appFrequencyMoments}

Reversibility is an intrinsic property of quantum computing that we often used in this paper.
It is known that any deterministic computation can be made reversible, and therefore implemented by a unitary map with a limited overhead on the time and space complexities~\cite{Ben89}. Nonetheless, implementing the reverse computation of a streaming algorithm would require processing the same stream but in the \emph{reverse} direction, which may not be always possible. This motivates our specific notion of \emph{reversible streaming algorithms}. We say that a streaming algorithm $\astr$ with memory size $M$ is \emph{reversible} if there exists a streaming algorithm $\astr^{-1}$ with memory size $M$ such that each computational steps of $\astr$ and $\astr^{-1}$ are reversible, and in addition each pass of $\astr$ can be undone by one pass of $\astr^{-1}$ in the \emph{same} direction.

Even if it is not clear how to make any streaming algorithm reversible, it is sufficient for our purpose to show how to achieve it when the streaming algorithm is a \emph{linear sketch}.

\begin{defn}
  \label{Def:linSketch}
  We say that a (one-pass) streaming algorithm $\astr$ is a \emph{linear sketch algorithm} with memory $M$, update time $T_{\mathit{upd}}$ and reconstruction time $T_{\mathit{rec}}$ if there exists a family $\{L_r\}_{r \in \rn^M}$ of linear functions $L_r : \R^n \ra \R^M$, and two deterministic algorithms $\alg_{\mathit{upd}}$ and $\alg_{\mathit{rec}}$ running in time $T_{\mathit{upd}}$ and $T_{\mathit{rec}}$ (respectively) and space $\bo{M}$, such that $\astr$ behaves as follows:
    \begin{enumerate}
      \item Draw $r \in \rn^M$ uniformly at random and store it in memory. Initialize $\mathtt{L} = 0$.
      \item Given $u_j = (i,\lambda)$, apply $\alg_{\mathit{upd}}$ on input $r$, $u_j$ to compute $L_r(\lambda e_i)$ and update $\mathtt{L} \leftarrow \mathtt{L} + L_r(\lambda e_i)$
      \item At the end of the stream, apply $\alg_{\mathit{rec}}$ on input $r$, $\mathtt{L}$ to compute the output of the algorithm
    \end{enumerate}
\end{defn}

Observe that, by linearity of $L_r$, the value of $\mathtt{L}$ in Definition \ref{Def:linSketch} after the $j$-th item has been processed is $\mathtt{L} = L_r(x(j))$. Linear sketch algorithms play an important role in the turnstile model, since they can implement essentially \emph{all} streaming algorithms \cite{LNW14,AHLW16}. Moreover, they are highly parallelizable, which facilitates their adaptation to the multi-pass model. In addition they can be made reversible as proved below. This property stems from the fact that the content of the memory, at any step of the computation, is unchanged under any permutation of the order of arrival of the updates received so far (because of the linearity of $L_r$).

\begin{prop}
  \label{Prop:reversibleSketch}
  For any linear sketch algorithm $\astr$ with parameters $(M,T_{\mathit{upd}},T_{\mathit{rec}})$, there exists a \emph{reversible} streaming algorithm $\mathcal{R}(\astr)$ with memory size $\bo{M \cdot \log\left(T_{\mathit{upd}} \cdot T_{\mathit{rec}}\right)}$ that computes the same output as $\astr$.
\end{prop}

\begin{proof}
  First we observe from~\cite{Ben89} that any (non-streaming) classical algorithm $\alg$ can be turned into a reversible one $\mathcal{R}(\alg)$ that computes the same output as $\alg$, performs $T^2$ computation steps and uses $\bo{M \log T}$ memory cells.

  We assume that the random seed $r \in \rn^M$ is pre-loaded in memory. Algorithm $\mathcal{R}(\astr)$ is implemented as follows. For each update $u(j) = (i,\lambda)$, use algorithm $\mathcal{R}(\alg_{\mathit{upd}})$  to compute reversibly $L_r(\lambda e_i)$, copy the result to $\mathtt{L} \leftarrow \mathtt{L} + L_r(\lambda e_i)$, and undo the computation of $L_r(\lambda e_i)$ with $\mathcal{R}(\alg_{\mathit{upd}})^{-1}$. The reconstruction part is done at the end of the stream using $\mathcal{R}(\alg_{\mathit{rec}})$.

  The reverse algorithm $\mathcal{R}(\astr)^{-1}$ first uncomputes the reconstruction part using $\mathcal{R}(\alg_{\mathit{rec}})^{-1}$. Then, for each update $u(j) = (i,\lambda)$, it computes $L_r(\lambda e_i)$ with $\mathcal{R}(\alg_{\mathit{upd}})$, updates $\mathtt{L} \leftarrow \mathtt{L} - L_r(\lambda e_i)$, and uncomputes $L_r(\lambda e_i)$ using $\mathcal{R}(\alg_{\mathit{upd}})^{-1}$.
\end{proof}

  \section{Approximating graph parameters in the query model}
  \label{App:appGraphParameters}

We fix a few notations that are used in the next two sections.

\begin{nota}
  Let $G = (V,E)$ be a graph, where $V = [n]$ for some integer $n$. For each vertex $v \in V$, we let $N_v$ equal the set of neighbor vertices to $v$, $E_v$ the set of edges adjacent to $v$, and $d_v = |N_v| = |E_v|$ the degree of $v$. Similarly, $T_v$ is the set of triangles adjacent to $v$, and $t_v = |T_v|$ its cardinality. We define the total order $\prec$ on $V = [n]$ where $u \prec v$ if $d_u < d_v$, or $d_u = d_v$ and $u < v$ (where $<$ is the natural order on $[n]$). We let $d_v^+$ equal the number of neighbors $w$ of $v$ such that $d_v \prec d_w$.
\end{nota}

\begin{fact}
  \label{Fact:dv+}
  For all vertex $v \in V$, we have $d_v^+ \leq \sqrt{2m}$.
\end{fact}


We will also use the following combination of Theorems \ref{Thm:EpsApproxDec} and \ref{Thm:VarEpsApprox}.

\begin{thm}
\label{Thm:VarEpsApproxDec}
  There is an algorithm that takes as input a variable-time sampler $\samp$, a function $f : x \mapsto A/x^{\alpha}$ for some reals $A, \alpha > 0$, two values $0 < \lo < \hi$, a real $\talgs \geq 1$, and two reals $0 < \eps, \delta < 1$ with $\delta < 2^{-2\alpha}$. If $f(\mus) \geq \phis / \mus$, $\talgs \geq \talgs(\samp)$ and $\lo \leq \mus < \hi$, this algorithm outputs an estimate $\muts$ that satisfies $|\muts - \mus| \leq \eps \mus$ with probability $1-\delta$, and it uses
    \[\so{f(\max(L,\mus)) \cdot \talgs \cdot \eps^{-2} \log^4(\tmax(\samp)) \log\left(\frac{\hi}{\lo}\right) \log\left(\frac{1}{\delta}\right)}\]
  \qss in expectation (both for the $\ell_1$ and $\ell_2$ average).
\end{thm}


\subsection{Approximating the number of edges}
\label{Sec:edge}

We show how to approximate the number $m$ of edges with $\so{n^{1/2} / (\eps m^{1/4})}$ quantum queries in expectation. We need the following estimator from Seshadhri \cite{Ses15}.

\boxest{Est:edge}{number $m$ of edges in a graph $G=(V,E)$ (from \cite{Ses15}).}{
{\bf Input:} query access to a graph $G = (V,E)$. \\
{\bf Output:} an estimate of $m = |E|$.

\begin{enumerate}
  \item Sample $v \in V$ uniformly at random. Sample $w \in N_v$ uniformly at random.
  \item If $v \prec w$, \underline{output} $n d_v$, else \underline{output} $0$.
\end{enumerate}
}

\begin{prop}
  \label{Prop:edge}
  If we let $X$ denote the output random variable of Estimator \ref{Est:edge}, then $\esp{X} = m$ and $\esp{X^2} \leq 2 \sqrt{2} n m^{3/2}$.
\end{prop}

\begin{proof}
  On the one hand, $\esp{X} = n^{-1} \sum_v (d_v^+ / d_v) \cdot n d_v = \sum_v d_v^+ = m$. On the other hand, $\esp{X^2} = n \sum_v d_v^+ \cdot d_v \leq 2 \sqrt{2} n m^{3/2}$, where we used Fact \ref{Fact:dv+}.
\end{proof}

We can now prove Theorem~\ref{Thm:edge}.

\begin{proof}[Proof of Theorem~\ref{Thm:edge}]
  We can implement Estimator \ref{Est:edge} with a sampler $\samp$ that computes, in constant time,
    \[\samp (\ket{0} \ket{0}) = \sum_{v \in V} \sum_{w \in N_v} \ket{v} \ket{w} \ket{\lambda(v,w)}\]
  where $\lambda(v,w) = n d_v$ if $v \prec w$, and $\lambda(v,w) = 0$ otherwise. According to Proposition \ref{Prop:edge}, we have $\mus = m$ and $\phis/\mus \leq 8^{1/4} n^{1/2}/m^{1/4}$. Consequently, using Corollary \ref{Cor:EpsApproxDecSimpl} with $f : x \mapsto 8^{1/4} n^{1/2}/x^{1/4}$, $L = 1$, $H = n^2$ and $\delta = 1/3$, we can estimate $\mt$ with accuracy $\eps$ and success probability $2/3$ using $\so{ \frac{n^{1/2}}{\eps m^{1/4}}}$ \qss in expectation.
\end{proof}

\paragraph{Lower bound}
We obtain a nearly matching lower bound by using a reduction from the two-player communication problem $\disj$. The proof is based on a construction from \cite{ER18}.

\begin{proof}[Proof of Theorem \ref{Thm:edgeLB}]
  Fix $n$, $m$, $\eps < 1/4$. Given an instance $(x,y) \in \rn^N \times \rn^N$ of size $N = n/(2\sqrt{4\eps m})$ for $\disj$, we construct a graph $G_{x,y}$ on $n$ vertices such that
    \[\begin{cases}
        \disj(x,y) = 1 \ \Longleftrightarrow \ G_{x,y} \ \text{has exactly $m$ edges} \\
        \disj(x,y) = 0 \ \Longleftrightarrow \ G_{x,y} \ \text{has at least $(1+4\eps)m$ edges}.
    \end{cases}\]
  The construction is as follows (see \cite[Section 4.1]{ER18}): fix any graph $H$ with $n/2$ vertices and $m$ edges, use half of the $n$ vertices in $G_{x,y}$ to construct a subgraph isomorphic to $H$, and partition the remaining $n/2$ vertices into $N$ sets $K_1, \dots, K_N$ of size $\sqrt{4\eps m}$. If $x_j = y_j = 1$ then $K_j$ is a clique, otherwise it is a set of isolated vertices. It is clear that at least one $K_j$ is a clique if and only if $\disj(x,y) = 0$.

  Consider now an algorithm that approximates with relative error $\eps$ the number of edges in any graph $G$ with $n$ vertices and $m$ edges using at most $Q$ quantum queries. Using the reduction above, it can be used on input $G_{x,y}$ to deduce the value of $\disj(x,y)$. We show how to implement it into a communication protocol of cost $\bo{Q \log n}$ on input $(x,y)$, using a standard technique from \cite{BCW98}. Alice runs the $Q$-query algorithm for $G_{x,y}$. When there is a vertex-pair query, her state is in a superposition $\sum_{v,w,b} \alpha_{v,w} \ket{v,w} \ket{b} \ket{\phi_{v,w}}$ over all pair of vertices $(v,w)$ in $G_{x,y}$. She has to compute $\sum_{v,w,b} \alpha_{v,w} \ket{v,w} \ket{b \oplus e_{v,w}} \ket{\phi_{v,w}}$ where $e_{v,w} = 1$ if and only if there is an edge between $v$ and $w$. If $(v,w)$ is an edge from the subgraph isomorphic to $H$, she can map directly $\ket{v,w} \ket{b} \mapsto \ket{v,w} \ket{b \oplus 1}$. If $v$ and $w$ belong to a same $K_j$, she appends $\ket{0}$ to $\ket{v,w} \ket{b}$, computes $\ket{v,w} \ket{b} \ket{0} \mapsto \ket{v,w} \ket{b} \ket{x_j}$, and sends the three registers to Bob. Then, Bob computes $\ket{v,w} \ket{b} \ket{x_j} \mapsto \ket{v,w} \ket{b \oplus (x_j \cdot y_j)} \ket{x_j} = \ket{v,w} \ket{b \oplus e_{v,w)}} \ket{x_j}$ and sends the result back to Alice who maps $\ket{v,w} \ket{b \oplus e_{v,w}} \ket{x_j} \mapsto \ket{v,w} \ket{b \oplus e_{v,w)}} \ket{0}$ to obtain the desired result. The degree and neighbor queries are implemented similarly. Each query requires $\bo{\log n}$ qubits of communication, hence the total communication cost is $\bo{Q \log n}$. Since the quantum communication complexity of any protocol computing $\disj$ must be $\om{\sqrt{N}}$ \cite{Raz03}, we obtain that $Q = \Omega(\sqrt{N}/ \log n) = \om{\frac{n^{1/2}}{(\eps m)^{1/4}} \cdot \log^{-1}(n)}$.
\end{proof}


\subsection{Approximating the number of triangles}
\label{Sec:triangle}

We show how to approximate the number $t$ of triangles with $\so{\frac{\sqrt{n}}{t^{1/6}} + \frac{m^{3/4}}{\sqrt{t}}}$ quantum queries in expectation. In order to keep this section concise, we describe an algorithm that computes a $(4/5 + \eps)$-approximation of $t$, though it is possible to obtain an $\eps$-approximation with similar ideas.

We begin with a simple estimator from \cite{ELRS17} for approximating the number $t_v$ of triangles adjacent to a given vertex $v \in V$.

%

\boxest{Est:classTv}{ratio of the number of adjacent triangles $t_v$ to the degree $d_v$ of a vertex $v$ (from \cite{ELRS17}).}{
{\bf Input:} query access to a graph $G = (V,E)$, a vertex $v \in V$. \\
{\bf Output:} an estimate of $t_v/d_v$.

\begin{enumerate}
  \item Sample $e \in E_v$ uniformly at random. Let $w$ be the endpoint of $e$ that is not $v$. Let $u$ be the smaller endpoint of $e$ according to $\prec$.
  \item If $d_u \leq \sqrt{2m}$, set $r = 1$ with probability $d_u/\sqrt{2m}$, \underline{output} $0$ otherwise. If $d_u > \sqrt{2m}$, set $r = \lceil d_u/\sqrt{2m}\rceil$.
  \item For $i = 1,\dots,r$:
   \begin{enumerate}
      \item Pick a neighbor $x$ of $u$ uniformly at random.
      \item If $e$ and $x$ form a triangle and $w \prec x$, set $X_i = \max(d_u,\sqrt{2m})$. Else, set $X_i = 0$.
    \end{enumerate}
  \item \underline{Output} $\frac{1}{r} \sum_{i = 1}^r X_i$.
\end{enumerate}}

\begin{prop}
  \label{Prop:classTv}
  If we let $X$ denote the output random variable of Estimator \ref{Est:classTv}, then $\esp{X} = t_v/d_v$ and $\var{X} \leq 2\sqrt{2m} t_v/d_v$. Moreover, the $\ell_2$-average running time of Estimator \ref{Est:classTv} is $\bo{1}$.
\end{prop}

\begin{proof}
  For each edge $e=(v,w)$, we let $t_{e,v}$ be the number of triangles $(v,w,x)$ such that $w \prec x$. It is clear that $t_v = \sum_{e \in E_v} t_{e,v}$. Moreover, $t_{e,v} \leq \sqrt{2m}$. Indeed, either $d_w \leq \sqrt{2m}$ (and thus  $t_{e,v} \leq d_w \leq \sqrt{2m}$), or $d_w > \sqrt{2m}$ and in this case $w$ cannot have more than $\sqrt{2m}$ neighbors of degree at least $\sqrt{2m}$.

  We first compute the mean of $X$ conditionned on the edge $e$ chosen at Step 1 and the value taken by $d_u$. We have $\esp{X | e,d_u \leq \sqrt{2m}} = (d_u/\sqrt{2m}) \cdot (t_{e,v}/d_u) \cdot \sqrt{2m} = t_{e,v}$ and $\esp{X | e,d_u > \sqrt{2m}} = (t_{e,v}/d_u) \cdot d_u = t_{e,v}$. Consequently, $\esp{X} = \frac{1}{d_v} \sum_{e \in E_v} \esp{X | e} = t_v/d_v$. Similarly, $\var{X^2 | e,d_u \leq \sqrt{2m}} \leq \esp{X^2 | e,d_u \leq \sqrt{2m}} = \sqrt{2m} t_{e,v}$ and $\var{X^2 | e,d_u > \sqrt{2m}} \leq (\sqrt{2m}/d_u) \cdot \esp{X_i^2 | e,d_u > \sqrt{2m}} \leq (\sqrt{2m}/d_u) \cdot (t_{e,v}/d_u) \cdot d_u^2 = \sqrt{2m} t_{e,v}$. Thus, using the low of total variance, $\var{X} \leq \frac{1}{d_v} \sum_{e \in E_v} (\sqrt{2m} t_{e,v} + t_{e,v}^2)$. Since $t_{e,v} \leq \sqrt{2m}$, it implies $\var{X} \leq 2\sqrt{2m} t_v/d_v$. Finally, the $\ell_2$-average running time of Step 3 is $\frac{1}{d_v} \sum_{w \in N_v} \left(\frac{\min(d_v,d_w)}{\sqrt{2m}}\right)^2 \leq \frac{1}{2m d_v} \sum_{w \in N_v} d_vd_w \leq \bo{1}$. The other steps of the estimator run in constant time.
\end{proof}

\begin{prop}
  \label{Prop:quantTv}
  There is a quantum algorithm that, given query access to any $n$-vertex graph $G$ with $m$ edges, a vertex $v \in V$, an integer $L$, an approximation parameter $\eps < 1$ and a failure parameter $\delta < 2^{-1}$, outputs an estimate $\tti_v$ of the number $t_v$ of triangles adjacent to $v$. If $L \leq t_v$, this estimate satisfies $|\tti_v - t_v| \leq \eps t_v$ with probability $1-\delta$. Moreover, for any $L$, it satisfies $\tti_v  \leq 2 t_v$ with probability $1-\delta$. The $\ell_2$-average running time of this algorithm, including its number of queries, is
    $\so{\left(1 + \frac{m^{1/4} \sqrt{d_v}}{\eps^2 \sqrt{L}}\right) \cdot \log(1/\delta)}$.
\end{prop}

\begin{proof}
  It is straightforward to implement Estimator \ref{Est:classTv} with a quantum sampler $\samp$, in a similar way as we did in the proof of Theorem \ref{Thm:edge}. This sampler satisfies $\mus = t_v/d_v$ and $\phis/\mus \leq 1 +(8m)^{1/4} \sqrt{d_v/t_v}$ according to Proposition \ref{Prop:classTv}. Moreover, its \emph{$\ell_2$-average} running time is $\talgs(\samp) = \bo{1}$. We estimate $t_v$ by applying Theorem \ref{Thm:VarEpsApproxDec} on $\samp$ with $f : x \mapsto 1 + (cm)^{1/4} \sqrt{d_v/x}$ (for a small enough constant $c$), $L' = L/d_v$ and $H = n^2$. The $\ell_2$-average running time of this algorithm is $\so{\left(1 + \frac{m^{1/4} \sqrt{d_v}}{\eps^2 \sqrt{L}}\right) \cdot \log(1/\delta)}$.
\end{proof}

The remaining part of our algorithm diverges from the approach taken in \cite{ELRS17}, that requires to set up a data structure for sampling edges uniformly in $G$. This technique seems to be an obstacle for improving the term $\bo{m^{3/2}/t}$ in the complexity. We circumvent this problem by combining \cite{ELRS17} with a bucketing approach from \cite{ELR15}, that partitions the graph's vertices into $k + 1 = \bo{\log n}$ buckets $B_0,\dots,B_k$, where
  \[ B_i = \{v \in V : t_v \in [(1+c)^{i-1},(1+c)^i] \} \]
for a small value $0 < c < 1$ to be chosen later. If we estimate the size $b_i = |B_i|$ of each bucket, then we would obtain an approximation of $\frac{1}{3} \sum_i |B_i| \cdot (1+c)^i \in [t,(1+c)t]$. We first show that the smallest sizes $|B_i|$ can be discarded, at the cost of a certain factor in the approximation.

\begin{lem}
  \label{Lem:discard}
  If $\bucket \subseteq \{0,\dots,k\}$ denotes the set of indices $i$ such that $|B_i| \geq \frac{(ct)^{1/3}}{k+1}$ and $|B_i| \geq \frac{ct}{(k+1) (1+c)^i}$, then
    \[ \frac{(1 - 2c)}{3} t \leq \frac{1}{3} \sum_{i \in \bucket} |B_i| \cdot (1+c)^i \leq (1+c) t \]
\end{lem}

\begin{proof}
  Define $B(v)$ to be the bucket that $v \in V$ belongs to, and let $V_{bad,1} = \left\{v \in V : |B(v)| < \frac{(ct)^{1/3}}{k+1} \right\}$ and $V_{bad,2} = \left\{v \in V : |B(v)| < \frac{ct}{(k+1) (1+c)^i} \right\}$. There are at most $(ct)^{1/3}$ vertices in $V_{bad,1}$. Consequently, at most $ct$ triangles have their three endpoints in $V_{bad}$. It implies $\sum_{v \in V_{bad,1}} t_v < 3ct + 2(1-c)t$. On the other hand, we have $\sum_{v \in V_{bad,2}} t_v \leq \sum_{i : |B_i| < \frac{ct}{(k+1) (1+c)^i}} |B_i| \cdot (1+c)^i < ct$. Consequently, $\frac{1}{3} \sum_{i \in \bucket} |B_i| \cdot (1+c)^i \geq t - \frac{1}{3} \sum_{v \in V_{bad,1} \cup V_{bad,2}} t_v > \frac{1}{3}(1 - 2c) t$.
\end{proof}

We are now ready to state the main result of this section.
\begin{thm}
  \label{Thm:triangleConst}
  There is a quantum algorithm that, given query access to an $n$-vertex graph $G$ with $m$ edges and an approximation parameter $\eps < 1$, outputs an estimate $\tti$ of the number $t$ of triangles of $G$ such that
    $|\tti-t| \leq (4/5 + \eps) t$
  with probability $2/3$. This algorithm performs
    $\so{\left(\frac{\sqrt{n}}{t^{1/6}} + \frac{m^{3/4}}{\sqrt{t}}\right) \cdot poly(1/\eps)}$ queries in expectation.
\end{thm}

\begin{proof}[Sketch of the proof]
  In the following, we assume that the threshold values $\frac{(ct)^{1/3}}{k+1}$ and $\frac{ct}{(k+1) (1+c)^i}$ used to define $\bucket$ are known, although $t$ is part of their definitions. In fact, it is easy to see that if $t$ is replaced with any value $\bar{t}$ in these expressions then the output of the algorithm described below will likely be smaller than $\bar{t}$ when $\bar{t} > 20 t$, and it will likely be larger than $\bar{t}$ when $\bar{t} < t / 20$. Thus, it suffices to perform a logarithmic search on $\bar{t}$ (starting with $\bar{t} = n^3$) to approximate the right threshold values.

  The general appropach of the algorithm is to compute separately an estimate $\bt_i$ of the size of each $B_i$ for $i \in \bucket$, and then to recombine them into $\sum_{i \in \bucket} \bt_i \cdot (1+c)^i$. If we had access to an oracle that returns $t_v$ for each $v \in V$, then it would suffice to perform order of $\sqrt{n/|B_i|}$ quantum queries for estimating $|B_i|$. Instead, we use the algorithm of Proposition \ref{Prop:quantTv} with threshold $L = (1+c)^{i-1}$ to decide if $v \in B_i$. Since we cannot distinguish efficiently $v \in B_i$ from $v \in B_{i+1}$ when $t_v$ is close to $(1+c)^i$, we are estimating a value between $|B_i|$ and $|B_{i-1}| + |B_i| + |B_{i+1}|$ instead. This adds a factor of $(1+c)^{-1} + 1 + (1+c) \leq 3 + c$ to the final approximation.

  In more details, we assign $v \in V$ to bucket $B_i$ if the output $\tti_v$ of the algorithm of Proposition \ref{Prop:quantTv} with input $v$, $L = (1+c)^{i-1}$, $\eps' = c/2$, $\delta = \eps/\poly(n)$ satisfies $\tti_v \in [(1+c)^{i-1},(1+c)^i]$. We apply this algorithm on a superposition over all vertices $v \in V$ to obtain a quantum sampler $\samp_i(\ket{0}\ket{0}) = n^{-1} \sum_{v \in V} \ket{v} \ket{\psi_v} \ket{e_v}$ over $\Omega = \rn$, where $\ket{\psi_v}$ is some garbage state, and $\ket{e_v}$ is a one-qubit state that equals $\ket{1}$ to indicate $v \in B_i$, and $\ket{0}$ otherwise. This sampler implements a Bernoulli distribution of mean $\mus \in \left[(1-\eps/8)|B_i|,(1+\eps/8)(3+c)|B_i|\right]$ (the $\eps/8$ error comes from the fact that the algorithm of Proposition \ref{Prop:quantTv} has probability $\delta = \eps/\poly(n)$ to fail).

  According to Proposition \ref{Prop:quantTv}, the $\ell_2$-average running time to compute each $\ket{\psi_v} \ket{e_v}$ is of the order of $\so{\left(1+\frac{m^{1/4}\sqrt{d_v}}{\eps^2 \sqrt{(1+c)^{i-1}}}\right) \log\left(\frac{n}{\eps}\right)}$. Thus, the $\ell_2$-average running time of $\samp_i$ is
    \[\so{\left(1 + \sqrt{\frac{1}{n} \sum_{v \in V} \left(\frac{m^{1/4} \sqrt{d_v}}{\eps^2 \sqrt{(1+c)^{i-1}}}\right)^2}\right) \log\left(\frac{n}{\eps}\right)}
    = \so{\left(1+\frac{m^{3/4}}{\eps^2 \sqrt{n(1+c)^{i-1}}}\right)\log\left(\frac{n}{\eps}\right)}\]

  We apply the algorithm of Theorem \ref{Thm:VarEpsApprox} on input $\samp_i$, $\Delta_{\samp_i} = \sqrt{n / \max\left(\frac{(ct)^{1/3}}{k+1},\frac{ct}{(k+1) (1+c)^i}\right)}$, $\hi = n$, $\lo = 1$, $\talgs = \so{\left(1+\frac{m^{3/4}}{\eps^2 \sqrt{n(1+c)^{i-1}}}\right)\log\left(\frac{n}{\eps}\right)}$, $\eps' = \eps/8$ and $\delta = \bo{1/\log(n)}$ to obtain an estimate $\bt_i \in \left[(1-\eps/8)^2|B_i|,(1+\eps/8)^2(3+c)|B_i|\right]$, in time
    \[\so{\sqrt{\frac{n}{\max\left(\frac{(ct)^{1/3}}{k+1},\frac{ct}{(k+1) (1+c)^i}\right)}}
          \left(1 + \frac{m^{3/4}}{\sqrt{n (1+c)^i}}\right) \cdot \poly(1/\eps)}
    = \so{\left(\frac{\sqrt{n}}{t^{1/6}} + \frac{m^{3/4}}{\sqrt{t}}\right) \cdot \poly(1/\eps)}\]

  Finally, we choose $c = \eps/4$ to define the buckets' width, which implies $\frac{1}{3} \sum_{i \in \bucket} |B_i| \cdot (1+c)^i \in [\frac{1}{3}(1-\eps/2)t,(1+\eps/4)t]$ according to Lemma \ref{Lem:discard}, and $\bt_i \in [(1 - \eps/4) |B_i|, 3(1+\eps/4) |B_i|]$ with large probability. Thus, $\frac{1}{3} \sum_{i \in \bucket} \bt_i \cdot (1+c)^i \in [\frac{1}{3}(1-\eps)t,3(1+\eps)t]$. Consequently, for $\tti = \frac{1}{5} \sum_{i \in \bucket} \bt_i \cdot (1+c)^i$, we have $|\tti - t| \leq (4/5+\eps) t$ with large probability.
\end{proof}

The approximation factor can be improved from $(4/5 + \eps)$ to $\eps$, by using a refined algorithm that combines techniques from \cite{ELR15} and \cite{ELRS17}. The first main idea is to randomly perturbate the buckets' boundaries (see \cite[Section 3.3.1]{ELR15}) to ensure that few vertices are close to them (this removes the previous factor $3(1+c)$ in the approximation). The second main idea is to modified the estimator used in Proposition \ref{Prop:quantTv} to compensate the loss introduced by discarding the buckets outside of $\bucket$. This leads to Theorem~\ref{Thm:triangleEps}.

\paragraph{Lower bound}
A nearly matching lower bound can be obtained with the same method as in Theorem \ref{Thm:edgeLB}, using the constructions given in Sections 4.1 and 4.3 of \cite{ER18} for the reduction to $\disj$. This leads to Theorem \ref{Thm:triangleLB}.


\end{document}